\documentclass[10pt,twocolumn,twoside]{IEEEtran}

\usepackage{subfigure}
\usepackage{citesort}
\usepackage[fleqn]{amsmath}
\usepackage{ascmac}
\usepackage{citesort}
\usepackage[dvips]{graphicx}
\usepackage{psfrag}
\usepackage{amsmath}
\usepackage{mathrsfs}
\usepackage{amsbsy}
\usepackage{amssymb}
\usepackage{latexsym}
\usepackage{color}
\usepackage{theorem}

\graphicspath{{./images/}}

\newtheorem{proposition}{Proposition}
\newtheorem{example}{Example}
\newtheorem{algorithm}{Algorithm}
\newtheorem{definition}{Definition}

\newtheorem{lemma}{Lemma}
\newtheorem{corollary}{Corollary}

\newtheorem{remark}{Remark}

\newcommand{\migip}{\hspace*{\fill} $\blacksquare$}

\newcommand{\euclidspace}{{\mathcal{H}}}
\newcommand{\opA}{{\mathscr{A}}}
\newcommand{\opD}{D}
\newcommand{\opM}{M_2}

\newcommand{\newz}{z}
\newcommand{\newv}{v}
\newcommand{\opL}{\mathscr{L}}
\newcommand{\genopL}{\mathfrak{L}}
\newcommand{\ball}{\mathcal{B}}
\newcommand{\inputA}{A}
\newcommand{\inputa}{a}
\newcommand{\outputb}{y}
\newcommand{\affinec}{c_2}
\newcommand{\hilbertx}{{\mathcal{X}}}
\newcommand{\hilbertnewy}{{\mathcal{Y}}}
\newcommand{\hilbertnewz}{{\mathcal{Z}}}
\newcommand{\hilbertarbh}{{\mathcal{H}}}
\newcommand{\hilbertarbk}{{\mathcal{K}}}

\newcommand{\Natural}{{\mathbb N}}
\newcommand{\norm}[1]{\left\|#1\right\|}
\newcommand{\abs}[1]{\left|#1\right|}

\newcommand{\real}{{\mathbb R}}

\newcommand{\innerprod}[2]{\left\langle{#1},{#2}\right\rangle}

\newcommand{\argmin}{\operatornamewithlimits{argmin}}

\newcommand{\gzh}{{\Gamma_0(\mathcal{H})}}
\newcommand{\prox}{{\rm Prox}}
\newcommand{\dom}{{\rm dom}\hspace*{.2em}}
\newcommand{\range}{{\rm range}\hspace*{.2em}}

\newcommand{\interior}{{\rm int}\hspace*{.2em}}

\definecolor{darkgreen}{rgb}{0,.6,0}
\definecolor{medorange}{rgb}{0.7,0.3,0}
\definecolor{cyancyan}{rgb}{0.68, 0.92, 0.92}

\usepackage{version}	
\begin{document}
%
\title{Linearly-involved Moreau-Enhanced-over-Subspace Model:
Debiased Sparse Modeling and Stable Outlier-Robust Regression}

%
%

\author{Masahiro~{\sc Yukawa},~\IEEEmembership{Senior Member,~IEEE,}~
Hiroyuki~{\sc Kaneko},~\IEEEmembership{Member,~IEEE,}\\
Kyohei~{\sc Suzuki},~\IEEEmembership{Graduate Student Member,~IEEE,}~
Isao~{\sc Yamada},~\IEEEmembership{Fellow,~IEEE}
\thanks{Manuscript received XXX yy, 2010; revised XXX xx, 200x.
This work was partially supported by JSPS Grants-in-Aid 
(22H01492).
}
\thanks{M.~Yukawa, H.~Kaneko, and K.~Suzuki are with the Department of 
Electronics and Electrical Engineering,
Keio University, Japan.
Address: Hiyoshi 3-14-1 (25-404), Kohoku-ku, Yokohama, Kanagawa 223-8522, Japan
(e-mail: yukawa@elec.keio.ac.jp).
I.~Yamada is with 
the Department of Information and Communications Engineering, 
Tokyo Institute of Technology,
2-12-1-S3-60, O-okayama,
Meguro-ku, Tokyo 152-8550, Japan
(e-mail: isao@sp.ce.titech.ac.jp).\\
\copyright ~2023 IEEE. Personal use is permitted, but republication/redistribution requires IEEE permission.
See https://www.ieee.org/publications/rights/index.html for more information.\\
Digital Object Identifier: 10.1109/TSP.2023.3263724
}
}

\maketitle







\begin{abstract}

We present an efficient mathematical framework
to derive promising methods
that enjoy ``enhanced'' desirable properties.
The popular minimax concave penalty 
for sparse modeling subtracts, 
from the $\ell_1$ norm, its Moreau envelope,
inducing nearly unbiased estimates and thus
yielding considerable performance enhancements.
To extend it to underdetermined linear systems,
we propose {\em the projective minimax concave penalty},
which leads to ``enhanced'' sparseness over the input subspace.
We also present a promising regression method which
has an ``enhanced'' robustness and substantial stability
by distinguishing outlier and noise explicitly.
The proposed framework, named 
{\em the linearly-involved Moreau-enhanced-over-subspace
 (LiMES) model},
encompasses those two specific examples
as well as two others: stable principal component
pursuit and robust classification.
The LiMES function involved in the model
is an ``additively nonseparable'' weakly convex function,
while the `inner' objective function to define
the Moreau envelope is ``separable''.
This {\em mixed nature of separability and nonseparability}
allows an application of the LiMES model to the underdetermined case
with an efficient algorithmic implementation.
Two linear/affine operators play key roles in the model:
one corresponds to the projection mentioned above and
the other takes care of robust regression/classification.
A necessary and sufficient condition for convexity
of the smooth part of the objective function
is studied.
Numerical examples show the efficacy of LiMES
in applications to sparse modeling and robust regression.


\end{abstract}



%


\section{Introduction}\label{sec:intro}

The primal goal of this article is to present
a unified mathematical framework to derive promising methods
that enjoy ``enhanced'' desirable properties.
The main body is divided into two parts.
The first part concerns two specific tasks of signal processing.
Specifically, the first task is 
finding sparse solutions of underdetermined linear systems
with small biases, and we present a certain data-dependent penalty function
yielding ``enhanced'' sparseness.
The second task is a robust regression task
in the presence of sparse outliers and large Gaussian noise,
and we present an efficient formulation that leads to ``enhanced''
robustness and substantial stability.
The second part presents the proposed framework which contains
{\em the linearly-involved Moreau-enhanced-over-subspace (LiMES) model}
at its core.
The proposed framework covers the two methods studied in the first part
as well as many others, including two more examples presented in the second part.
The background of the study of the first part is presented below,
followed by the details of each part.

\subsection{Background}

Sparsity awareness and outlier robustness are two key aspects of
paramount importance in regression (linear estimation),
which has a wide range of applications in many fields including
signal processing and machine learning \cite{theodoridis_book20,zoubir_book18}.
The $\ell_1$ penalty and the $\ell_1$ loss, a.k.a.~the least absolute
deviation (LAD), are known to yield sparse solutions
\cite{elad10,foucart13} and outlier-robust estimates
\cite{huber_book,pesme20}, respectively,
as opposed to the squared $\ell_2$ norm
which has widely been used for the  Tikhonov regularization
or the squared errors.
The $\ell_1$ norm is a convex relaxation of the $\ell_0$ pseudo-norm
(which is a direct discrete measure of sparsity counting the number of
nonzero entries);
i.e., the $\ell_1$ norm is the largest convex minorant of $\ell_0$
in a vicinity of the origin.
For better relaxations/approximations to ameliorate the performance,
a plethora of nonconvex alternatives to the $\ell_1$ norm have been proposed
\cite{chartrand08,marjanovic12,shen18,yao18,Bwen18},
including the $\ell_p$ quasi-norm for $p\in(0,1)$
(e.g., \cite{chartrand07,yukawa16,jeong14}
among many others), capped $\ell_1$ \cite{zhang09},
log-sum function \cite{candes08}, minimax concave (MC) \cite{zhang},
smoothly clipped absolute deviation (SCAD) \cite{fan01}, 
continuous exact $\ell_0$ (CEL0) \cite{soubies15}, to name a few.
See also the survey paper \cite{wen18} for more references.
Among those penalties,
MC and SCAD are well known to be {\em weakly convex}; i.e.,
those functions become convex by adding a scaled squared $\ell_2$ norm.

The notion of ``convexity-preserving'' nonconvex penalties
using weakly convex functions can be found
in the literature \cite{blake87,nikolova99}.
The idea is to preserve overall convexity of 
the whole objective function 
by exploiting strong convexity of the other term(s);
cf.~difference of convex (DC) programming \cite{dinh86}.
See, e.g., \cite{parekh16,lanza19} for more recent advances.
In addition that
the weakly convex penalties induce sparsity
with small estimation biases,
the optimization problems involving a quadratic function 
and such weakly convex penalties can be solved by powerful
convex-analytic algorithms with convergence guarantee
to a global minimizer.
It is widely known that
the $\ell_p$ quasi-norm resides 
between the $\ell_0$ and $\ell_1$ norms.
It has been shown recently that
the (properly-normalized) MC penalty
bridges the $\ell_0$ and $\ell_1$ norms by a single parameter
\cite[Example 2]{abe_ip20}.
This, together with its nice experimental performances,
motivates us to focus on the MC penalty.
Let us consider the squared-error fidelity (the least square
loss) penalized by a weakly convex function in linear regression.
The overall convexity can be preserved
in the overdetermined case
by choosing the regularization parameter properly.
In some important applications including high dimensional data
analysis and compressed sensing, 
however,
the overall convexity {\em cannot} be preserved because
the number of measurements is much smaller than the number of variables.

To overcome this strict limitation, 
the generalized MC (GMC) penalty has been proposed \cite{selesnick}.
Based on the fact that the MC penalty can be expressed as 
a difference between the $\ell_1$ norm and its Moreau envelope,
GMC inserts a matrix-valued tuning parameter
in the quadratic term of the Moreau envelope.
We refer to the $\ell_1$ norm as ``the seed function'' of the GMC penalty.
The GMC penalty has been extended (i) from $\ell_1$ norm to a more
general convex seed-function satisfying certain mild conditions and
(ii) to a composition of linear operator \cite{Abe2,abe_ip20}.
The extended function is called {\em linearly involved generalized
Moreau enhanced (LiGME)} penalty \cite{abe_ip20},
covering the Moreau enhanced penalties for the nuclear norm and total
variation among many others.
The important property common to those generalized penalties is 
{\em nonseparability} even if its seed function is {\it additively separable};
i.e., those penalties are not necessarily expressed
as a sum of individual functions each of which depends solely on 
each variable.
Thanks to its nonseparability, GMC/LiGME can be applied to
underdetermined linear systems.
While it has rigorous theoretical backbones,
its use in robust regression has not been investigated so far.
Although a number of nonconvex loss functions have been proposed 
\cite{yan13,hohm15,yuan15,wen17,javaheri18,tzagkarakis19,yang19}
indeed as alternatives to the convex ones such as LAD or Huber's loss
\cite{huber_book,pesme20},
global optimality has not been discussed in those previous works.


\subsection{Contributions --- Part I}
There are three research questions that motivate the present study,
two of which are stated in this part.\\
{\em 
(Q1) What is a function that is maximally close to
the MC penalty while being able to possess overall convexity 
in underdetermined situations?}\\
 We would like to reserve such a region, as much as possible,
 on which the newly developing penalty coincides with the MC penalty.
In the underdetermined case,
the fidelity function is {\em not} strongly convex in the whole space.
Precisely, while it is strongly convex
on the subspace spanned by the input vectors,
it is ``flat'' (it has no strong convexity at all)
on its orthogonal complement.
This immediately implies that
the penalty function needs to be convex
on the orthogonal complement to preserve the overall convexity.
This simple observation is the key for our first contributions
summarized below.

\begin{itemize}
 \item We propose {\em the projective  minimax concave (PMC) penalty}
in which 
the projection operator onto the input subspace
is used to annihilate the Moreau enhancement effects on 
its orthogonal complement.
PMC reduces to the original MC penalty on the input
subspace while it reduces to the $\ell_1$ norm (a convex relaxation of
      the $\ell_0$ pseudo-norm) on its orthogonal complement
(see Proposition \ref{proposition:rmc_properties}).
This means that PMC gives an answer to the first question shown above.

\item The formulation involving PMC enhances sparsity with small
      estimation biases in the underdetermined case,
and thus it is referred to as {\em debiased sparse modeling}.

 \item While the PMC penalty itself is ``additively nonseparable'',
the ``internal'' objective function to define the Moreau envelope is 
``separable'' as long as the seed function is separable.
This {\em mixed nature of separability and nonseparability}
allows PMC to preserve overall convexity in the underdetermined case
with its efficient implementation using no extra variable.

\end{itemize}
{\em 
(Q2) Can we build a regression method that is highly robust
against huge outliers and stable even in severely noisy environments?}
\begin{itemize}
 \item We propose {\em stable outlier-robust regression (SORR)}
under the assumption that the noise is Gaussian and the outlier is sparse.
An additional variable vector is introduced to model the Gaussian noise
on top of the adoption of the MC-based fidelity function to evaluate
the sparse outliers, thereby reflecting the noise Gaussianity and 
the outlier sparsity in a reasonable way.

 \item SORR is a promising approach because
(a) it is highly robust and stable even in severely noisy environments,
and (b) it can be implemented efficiently 
by the operator splitting methods since
overall convexity of the whole cost is preserved under a certain condition.
This indicates that SORR resolves a certain intrinsic tradeoff existing 
in the conventional approaches (see Section \ref{subsubsec:sorr_motivation}).
\end{itemize}


\subsection{Contributions --- Part II}

The two methods proposed in the first part are based on weakly convex functions.
This gives rise to the third question.\\
{\em 
(Q3) Can we build a mathematical modeling framework 
to treat weakly convex functions in a unified fashion
for regression/classification tasks such as those studied in the first part?
}

\begin{itemize}
 \item 
We propose the LiMES model which encompasses the debiased sparse modeling and SORR 
as its particular examples.
The other examples of LiMES presented in this paper are
{\rm stable principal component pursuit (SPCP)}
\cite{zhou10} and robust classification.
For the latter application, in particular,
the popular hinge loss is enhanced by the LiMES model
with its expression 
as a composition of the support function
of a closed interval $[-1,0]$ and some affine operator.


\item 
A necessary and sufficient condition for the smooth part 
of the whole cost to be convex is presented under a certain assumption
(Proposition \ref{proposition:positive_definite_case_necessary_condition}).

\item 

The structure of LiMES admits its decomposition into a sum of
smooth and nonsmooth (proximable) convex functions,
allowing an application of the efficient operator splitting methods
to solve the posed problem.
The gradient of the smooth part produces an {\it implicit} proximity
operator, which contributes to reducing the estimation bias 
caused by the proximity operator appearing {\it explicitly} 
in the original form of the optimization algorithm.
\end{itemize}

Numerical examples show the efficacy of the LiMES framework.
Specifically, the PMC penalty achieves debiased sparse modeling 
for underdetermined systems as well as outperforming GMC,
and SORR\footnote{
Partial results (the SORR estimator and a special case of the LiMES
model) of this work
have been presented at a conference
\cite{yukawa_eusipco22} with no detailed discussions nor proofs for theoretical results.
} achieves stable and remarkably robust performances 
in the presence of both heavy Gaussian noise and
sparse outlier as well as outperforming the existing robust methods.


\subsection{Notation and mathematical tools}
\label{subsec:notation}

Let $\real$, $\real_{++}$, and $\Natural$ denote
the sets of real numbers, strictly positive real numbers, and
nonnegative integers, respectively.
Let $(\euclidspace,\innerprod{\cdot}{\cdot})$ be a real Hilbert space
equipped with inner product $\innerprod{\cdot}{\cdot}$, of which
the induced norm is denoted by $\norm{\cdot}$.
Throughout the paper, we focus on the finite dimensional case, although
many of the arguments given in this section apply to
the infinite dimensional case.
We denote by $I:\euclidspace\rightarrow\euclidspace$ the identity operator,
and by $0\in\euclidspace$ and $O:\euclidspace\rightarrow
\euclidspace:x\mapsto 0$ 
the zero vector of $\euclidspace$ and
the zero operator, respectively.
We may use the same notation of inner product, norm, zero vector, and
zero operator for other Hilbert spaces, whenever it causes no confusion.
A subset $C\subset \euclidspace$ is convex 
if $\alpha x+(1-\alpha)y \in C$ for all $(x,y,\alpha)\in C\times C\times [0,1]$.
Given a nonempty closed convex set $C\subset\euclidspace$,
the projection operator is defined by
$P_C:\euclidspace\rightarrow\euclidspace:x\mapsto \argmin_{y\in C}\norm{x-y}$.
An operator $T:\euclidspace\rightarrow\euclidspace$ is 
{\em Lipschitz continuous with constant} $L\in\real_{++}$  if 
$\norm{T(x) - T(y)} \leq L \norm{x-y}$ for every
$x,y\in\euclidspace$.
The projection operator $P_C$ is 
Lipschitz continuous with constant 1 (i.e., {\em nonexpansive}).

A function $f:\euclidspace\rightarrow
(-\infty,+\infty]:=\real\cup\{+\infty\}$ is convex on $\euclidspace$ if
$f(\alpha x + (1-\alpha)y)\leq 
\alpha f(x) + (1-\alpha)f(y)$ for all
$(x,y,\alpha)\in\dom f\times\dom f\times [0,1]$, where
$\dom f:= \{x\in\euclidspace\mid f(x)<+\infty\}$.
If in addition $\dom f\neq \emptyset$,
$f$ is a {\it proper convex} function.
For $\eta\in\real_{++}$,
$f$ is $\eta$-{\em strongly convex}
if $f- 0.5 \eta\norm{\cdot}^2$ is convex,
and it is $\eta$-{\em weakly convex}
if $f+ 0.5\eta\norm{\cdot}^2$ is convex.
A convex function $f:\euclidspace\rightarrow (-\infty,+\infty]$ is 
{\it lower semicontinuous} (or {\it closed}) on $\euclidspace$ 
if the level set
${\rm lev}_{\leq a} f:=
\left\{x\in\euclidspace: f(x)\leq a\right\}$
is closed for every $a\in\real$.
The set of all proper lower-semicontinuous convex functions
defined over $\euclidspace$
is denoted by $\Gamma_0(\euclidspace)$.
Given a function $f\in\gzh$, 
{\em the Fenchel conjugate of $f$} is 
$\gzh\ni f^*: x\mapsto
\sup_{y\in\euclidspace} \innerprod{x}{y} - f(y)$.
The Moreau envelope (smooth convex approximation)
of $f$ of index
$\gamma\in \real_{++}$ is defined
by \cite{moreau62,moreau65,combettes05}
\begin{align}
\hspace*{-2em} ^{\gamma} f:\euclidspace\rightarrow\real:
& \ x\mapsto \min_{y\in\euclidspace}
\left(
f(y) + 0.5\gamma^{-1}\norm{x-y}^2
\right)\nonumber\\
\hspace*{-2em} &\hspace*{-2.1em}
=f(\prox_{\gamma f}(x)) + 0.5\gamma^{-1}\norm{x-\prox_{\gamma f}(x)}^2,
\label{eq:moreau_envelope}
\end{align}
where
\begin{equation}
\hspace*{0em}\prox_{\gamma f}:
\euclidspace \!\rightarrow\! \euclidspace:
x \mapsto \argmin_{y \in \mathcal{H}}
\!\left(f (y) \!+\! 0.5 \gamma^{-1}
\norm{x \!- \!y}^2\right)
\label{eq:def_prox}
\end{equation}
is the proximity operator of $f$ of index $\gamma$.
The gradient of the Moreau envelope $^{\gamma} f$
is given by \cite{moreau62,moreau65,combettes05,yyy_springer_book11}
$\nabla \hspace*{.3em}^{\gamma}f= \gamma^{-1}
\left(I - {\rm Prox}_{\gamma f}\right)$, 
which is Lipschitz continuous with constant $\gamma^{-1}$.
The following identity holds in general \cite[Theorem 14.3]{combettes}:
\begin{equation}
 \hspace*{.3em}^{\gamma}f
+\hspace*{.3em}^{1/\gamma}(f^*)\circ \gamma^{-1}I=0.5\gamma^{-1} \norm{\cdot}^2.
\label{eq:moreau_identity}
\end{equation}

For any $n,m\in\Natural^*:=\Natural\setminus\{0\}$,
the $n\times n$ identity and zero matrices are denoted by
$I_n$ and $O_n$, respectively, and
the $n\times m$ zero matrix is denoted by $O_{n\times m}$.
Matrix transpose is denoted by $(\cdot)^\top$.
The $\ell_1$ and $\ell_2$ norms of Euclidean vector 
$x:=[x_1,x_2,\cdots,x_n]^\top\in\real^n$
are defined respectively by 
$\norm{x}_1:=\sum_{i=1}^n \abs{x_i}$ and 
$\norm{x}_2:= (\sum_{i=1}^n x_i^2)^{1/2}$.


\section{Two Novel Formulations for Linear Regression}
\label{sec:rmc_stable}

Two specific situations in linear regression are considered.
We first present the PMC penalty
to obtain debiased estimates for sparse modeling
under possibly underdetermined systems.
We then present SORR to combat the noise and outlier
in a separate fashion.
Given a coordinate system, a function is said to be
{\em ``additively separable''} when it is a superposition
of individual functions of each parameter.\footnote{Additive separability depends on the coordinate system.}
The $\ell_1$ norm is a simple example of separable functions.

\subsection{PMC penalty for debiased sparse modeling}
\label{subsec:sparse_regression}

\subsubsection{Sparse modeling}

We consider sparse modeling under
the standard linear model 
$\outputb:=Ax_{\diamond} + \varepsilon_{\star}$.
Here, $x_{\diamond} \in \mathbb{R}^{n}$ is the sparse unknown vector to be
estimated,
$\varepsilon_{\star} \in \mathbb{R}^{m}$ is the Gaussian noise vector,
and
$\inputA:=[\inputa_1~\inputa_2~\cdots 
\inputa_m]^\top\in\real^{m\times n}\setminus \{O_{m\times n}\}$ and
$\outputb:=[\outputb_1,\outputb_2,\cdots,\outputb_m]^\top\in\real^m$
are the input matrix 
and the output vector, respectively,
with the $i$th input vector
$\inputa_i\in\real^n$,
$i=1,2,\cdots,m$,
and its corresponding output 
$\outputb_i\in\real$.
The task is the following: find the sparse vector $x_{\diamond}\in\real^n$
given $A$ and $y$.
The linear system is supposed to be possibly {\em underdetermined}; 
i.e., $\inputA^\top \inputA\in \real^{n\times n}$ might be singular.

\subsubsection{The PMC penalty}
To reduce the estimation bias while preserving the overall convexity,
we propose the following formulation 
(which we refer to as {\em debiased sparse modeling\footnote{
It differs from {\em debiased lasso estimator} studied in statistics
 \cite{javanmard14} which ``desparsifies'' the
estimate to reduce the estimation bias by adding a Newton step to the lasso estimate.
}}):
\begin{equation}
\min_{x\in\real^n} ~
0.5\norm{\inputA x - \outputb}_2^2 + \mu
\underbrace{\left[\norm{x}_1 - \hspace*{.1em}^{\gamma}\norm{\cdot}_1(P_{\mathcal{M}}x)
\right]}_{\Phi_{\gamma}^{\rm PMC}(x)},
\label{eq:sparse_regression_mc_underdetermined}
\end{equation}
where $P_{\mathcal{M}}= A^{\dagger} A\in\real^{n\times n}$
is the orthogonal projection operator onto
$\mathcal{M}:=
{\rm null}^{\perp} \hspace*{.1em}A ~(=\range \inputA^\top)\subset\real^n$,
$\mu\in\real_{++}$ is the regularization parameter, and
\begin{equation}
\Phi_{\gamma}^{\rm PMC}(x):= \norm{x}_1 - \hspace*{.1em}^{\gamma}\norm{\cdot}_1(P_{\mathcal{M}}x)
\end{equation}
is the proposed PMC penalty.
Here, 
$(\cdot)^{\dagger}$ and $(\cdot)^{\perp}$ denote the Moore-Penrose pseudoinverse
and the orthogonal complement of subspace, respectively.

Using the identity \eqref{eq:moreau_identity},
the standard MC penalty \cite{zhang,selesnick} 
can be written as
$ \Phi_{\gamma}^{\rm MC} (x) := \norm{x}_1 -
\hspace*{.1em}^{\gamma}\norm{\cdot}_1(x)
=\norm{x}_1 + \hspace*{.1em}^{\gamma^{-1}}(\norm{\cdot}_1^*)(\gamma^{-1}x)
- 0.5\gamma^{-1} \norm{x}^2$.
Here, the subtraction of the Moreau envelope
$^{\gamma}\norm{\cdot}_1(x)$ from $\norm{x}_1$ leads to nearly unbiased
estimation \cite{zhang},
and it hence enhances the performance significantly.
As the conjugate function $\norm{\cdot}_1^*$ of
$\norm{\cdot}_1$ is convex, so is its Moreau envelope
$^{\gamma^{-1}}(\norm{\cdot}_1^*)$, and thus
$\Phi_{\gamma}^{\rm MC} (x)$ is $\gamma^{-1}$-weakly convex.
The MC penalty cannot therefore be applied to
the underdetermined case when $A^\top A$ is singular,
because
$0.5\norm{\inputA x - \outputb}_2^2+ \mu \Phi_{\gamma}^{\rm MC} (x)$ 
cannot be convex for any $\mu\in\real_{++}$.
Intuitively, 
the convexity of the fidelity term $0.5\norm{\inputA x - \outputb}_2^2$
cannot annihilate
the concavity of the negative quadratic term $- 0.5\gamma^{-1} \norm{x}^2$,
since the former function is flat (i.e., it possesses zero curvature)
over $\mathcal{M}^\perp(={\rm null}\hspace*{.2em} A)$, or any of its translations.
Here comes the idea of inserting $P_{\mathcal{M}}$ 
into the penalty in \eqref{eq:sparse_regression_mc_underdetermined}.
The projection operator $P_{\mathcal{M}}$ restricts
the concavity to
$\mathcal{M}$ $(={\rm null}^{\perp}
\hspace*{.1em}A)$,
on which the fidelity function is strongly convex,
so that the overall convexity can be preserved.
As a result, the Moreau enhancement effect
is restricted to $\mathcal{M}$ as well.
A formal discussion about the convexity issue 
is postponed to Section \ref{subsubsec:isd_method}.
In the overdetermined case,
PMC reduces to the standard MC penalty,
as $\mathcal{M}=\real^n$ and thus $P_\mathcal{M}=I$.

We mention that the PMC penalty $\Phi_{\gamma}^{\rm MC}$ depends 
on the input subspace $\mathcal{M}$.
This comes from a requirement for the preservation of overall convexity.
This design strategy also has a more positive aspect
in such specific situations when
the desired solution belongs to a known input subspace at least with high probability
(and possibly one is allowed to generate the input vectors so that it
spans that particular subspace).

\subsubsection{Properties of the PMC penalty}
\label{subsubsec:property_of_rmc}

Some properties of PMC are given below.

\begin{remark}[Separability and nonseparability]
\label{remark:separability_of_rmc}

The PMC penalty $\Phi_{\gamma}^{\rm PMC}$ in \eqref{eq:sparse_regression_mc_underdetermined}
is ``additively nonseparable''
as a function of $x$
with respect to the Cartesian coordinate system
(i.e., PMC is not represented as a sum of
individual functions of each component of $x$),
unless $P_{\mathcal{M}}$ is a diagonal matrix.
In contrast, the second term of $\Phi_{\gamma}^{\rm PMC}$
is given by
$^{\gamma}\norm{\cdot}_1(P_{\mathcal{M}}x)=
\min_{u\in\real^n} \left[\norm{u}_1+0.5\gamma^{-1}\norm{P_{\mathcal{M}}x
 -u}_2^2\right]
= \min_{u_1,u_2,\cdots,u_n\in\real} \sum_{i=1}^{n} \phi_i(u_i)$,
in which the objective function is ``separable'' as a function of $u$.
Here, $\phi_i(u_i):=\abs{u_i}+0.5\gamma^{-1}(p_i-u_i)^2$ with $p_i\in\real$
 denoting the $i$th component of $P_{\mathcal{M}}x$.
This mixed nature of separability and nonseparability
is crucial.
It is known indeed that, 
to preserve the overall convexity when $A^\top A$ is singular,
a nonconvex penalty needs to be nonseparable, excluding a trivial case
 {\rm \cite{selesnick16}}.
At the same time, thanks to the separability mentioned above,
the minimizer of the objective function
$\norm{\cdot}_1+0.5\gamma^{-1}\norm{P_{\mathcal{M}}x- \cdot}_2^2$ is
given simply by ${\rm soft}_{\gamma}(P_{\mathcal{M}}x)$.
Since $^{\gamma}\norm{\cdot}_1(P_{\mathcal{M}}x)$ is merely 
the composite of the linear operator $P_{\mathcal{M}}$ and
the Moreau envelope of the $\ell_1$ norm, 
an application of the chain rule 
with $P_{\mathcal{M}}^*=P_{\mathcal{M}}$ gives
the gradient
$\nabla(^{\gamma}\norm{\cdot}_1\circ P_{\mathcal{M}})(x)
= \gamma^{-1}P_{\mathcal{M}} \circ(I - \prox_{\gamma
 \norm{\cdot}_1})\circ P_{\mathcal{M}} (x)$,
where the gradient operator $\nabla(^{\gamma}\norm{\cdot}_1\circ P_{\mathcal{M}})$ is 
Lipschitz continuous with constant  $\gamma^{-1}$.
This smoothness property simplifies the optimization procedure, 
as shown in Section \ref{subsubsec:isd_method}.

\end{remark}

\begin{proposition}
\label{proposition:rmc_properties} 
For the PMC penalty, the following hold:
\begin{enumerate}
  \item [(a)]
The PMC penalty $\Phi_{\gamma}^{\rm PMC}$ coincides with the MC penalty
on the input subspace $\mathcal{M}$; i.e.,
$\Phi_{\gamma}^{\rm PMC} (x) = \Phi_{\gamma}^{\rm MC} (x) = \norm{x}_1 - \hspace*{.1em}^{\gamma}\norm{\cdot}_1(x)$
for $x\in \mathcal{M}$.

 \item [(b)]
The PMC penalty $\Phi_{\gamma}^{\rm PMC}$ is reduced to the $\ell_1$ norm
on the orthogonal complement $\mathcal{M}^{\perp} (={\rm null}
 \hspace*{.2em}A)$; i.e.,
 $\Phi_{\gamma}^{\rm PMC}(x) = \norm{x}_1$ for $x\in \mathcal{M}^{\perp}$.

\end{enumerate}
\end{proposition}
\begin{proof}
Clear from \eqref{eq:sparse_regression_mc_underdetermined}.
\end{proof}


\begin{remark}[PMC penalty bridges $\ell_0$ and $\ell_1$ over $\mathcal{M}$]
\label{remark:bridge}
An important implication of
Proposition \ref{proposition:rmc_properties}(a)
is then that the PMC penalty gives a bridge by a single parameter $\gamma$
between 
the direct measure $\norm{\cdot}_0$ of sparsity and
its convex relaxation $\norm{\cdot}_1$ on the subspace $\mathcal{M}$.
To be specific, we define
$\tilde{\Phi}_{\gamma}^{\rm PMC} :=
\theta_{\gamma} \Phi_{\gamma}^{\rm PMC}$,
where $\theta_{\gamma}:=
\begin{cases}
 \frac{2}{\gamma} & \mbox{if } \gamma\in(0,2)\\
1 & \mbox{if } \gamma\in [2,+\infty).
\end{cases}$
Then, it follows, for any $x\in \real^n(\supset \mathcal{M})$,
 that
(i)
$\lim_{\gamma\downarrow 0}\tilde{\Phi}_{\gamma}^{\rm  PMC}(x)=
 \norm{x}_0$
 {\rm \cite{abe_ip20}}
[see also Example \ref{example:alime_penalty}(a)], 
and (ii) $\lim_{\gamma\rightarrow +\infty}\tilde{\Phi}_{\gamma}^{\rm  PMC}(x)=
 \norm{x}_1$.
Here, the latter argument can be justified by  observing that
$0\leq\hspace*{.3em} ^{\gamma} \norm{\cdot}_1(x)=
\min_{u\in\real^n}\left(
\norm{u}_1 + 0.5\gamma^{-1}\norm{u-x}_2^2
\right)\leq
\norm{0}_1 + 0.5\gamma^{-1}\norm{0-x}_2^2\rightarrow 0
$ as $\gamma\rightarrow +\infty$
for any $x\in \real^n$.
\end{remark}

We emphasize that those remarkable properties given in Remarks
\ref{remark:separability_of_rmc} and \ref{remark:bridge} and Proposition
\ref{proposition:rmc_properties}
come from the ``structure'' of PMC (see Remark
\ref{remark:separability_of_rmc}),
not from the use of the projection operator.
We mention that PMC is non-monotonic, and
it is decreasing in some direction(s) so that
it may take negative values.
Although this may cause overestimation,
PMC tends to perform better than GMC (which would suffer from 
underestimation owing to a shrinking bias to a certain extent),
as shown by simulations in Section \ref{subsec:exp_sparse_regression}.
In fact, all those properties make PMC be significantly different from GMC and
its related works.

\subsubsection{Iterative shrinkage and debiasing algorithm}
\label{subsubsec:isd_method}

The problem in \eqref{eq:sparse_regression_mc_underdetermined}
can be viewed as 
\begin{equation}
 \min_{x\in\real^n}
\underbrace{0.5\norm{\inputA x - \outputb}_2^2 -\mu
\hspace*{.1em}^{\gamma}\norm{\cdot}_1(P_{\mathcal{M}}x)}_{\rm smooth}
+
\underbrace{\mu\norm{x}_1.}_{\rm nonsmooth}
\label{eq:sparse_regression_mc_underdetermined_smooth_nonsmooth}
\end{equation}
Since the gradient of the smooth part in
\eqref{eq:sparse_regression_mc_underdetermined_smooth_nonsmooth} and 
the proximity operator of the $\ell_1$ norm
are available (see Remark \ref{remark:separability_of_rmc}),
the proximal gradient method \cite{PFBS1,PFBS2}
can be applied,
under the convexity condition presented in Proposition
\ref{proposition:rmc_convexity} below,  directly to
\eqref{eq:sparse_regression_mc_underdetermined_smooth_nonsmooth}
to obtain the following algorithm.
Given an initial point $x_0\in\real^n$,
generate a sequence $(x_k)_{k\in\Natural}\subset \real^n$
by (cf.~Section \ref{subsec:twin_prox})
\begin{align}
\hspace*{0em} x_{k+1} := {\rm soft}_{\beta_k\mu} [x_k 
+&~ \beta_k \mu \gamma^{-1}\!
P_{\mathcal{M}} \left(x_k \!-\! {\rm soft}_\gamma (P_{\mathcal{M}}x_k)\right)
\nonumber\\
& \hspace*{.5em}- \beta_k \inputA^\top (\inputA x_k -\outputb) ],~k\in\Natural,
\label{eq:isda}
\end{align}
where
$\beta_k \in (0,2/(\lambda_{\max}(\inputA^\top\inputA) + \mu
\gamma^{-1}))$ is the step size, and,
for any $\delta\in\real_{++}$,
${\rm soft}_\delta:={\rm Prox}_{\delta
\norm{\cdot}_1}:\real^n\rightarrow\real^n:
x:=[x_1,x_2,\cdots,x_n]^\top \mapsto
[\varphi_{\delta}(x_1),\varphi_{\delta}(x_2),\cdots,
\varphi_{\delta}(x_n)]^\top$,
$n\in\Natural^*$,
is the shrinkage (soft thresholding) operator.
Here, $\lambda_{\max}(\cdot)$ denotes
the largest eigenvalue, and
$\varphi_{\delta}:\real\rightarrow \real:
a\mapsto {\rm sign}(a)\max\{0,\abs{a}-\delta\}$,
where ${\rm sign}(a):= 1$ if $a\geq 0$;
${\rm sign}(a):= -1$ otherwise.

\begin{remark}
The algorithm in \eqref{eq:isda} involves no auxiliary vector
thanks to the ``separability'' discussed in Remark \ref{remark:separability_of_rmc}.
This is in contrast to the GMC-based formulation {\rm \cite{selesnick}} for which auxiliary vectors
(with a saddle-point problem considered)
need to be used, because the objective function 
of the minimization problem involved in the generalized Huber function
(the generalized Moreau envelope) is typically ``nonseparable''.
This auxiliary-vector-free nature of PMC could potentially reduce
the memory requirement with respect to that for the  algorithm in 
{\rm \cite{selesnick}} applied to the GMC-based formulation.

\end{remark}

To understand the behaviour of the algorithm in \eqref{eq:isda}
geometrically, let us first consider the case when it is applied to
the original MC penalty; i.e., the case of $P_{\mathcal{M}}=I$.
In this case, the second term in the bracket of \eqref{eq:isda} reduces to
$\beta_k \mu \gamma^{-1}(x_k - {\rm soft}_\gamma(x_k))$, 
which actually reduces the shrinking bias
caused by the shrinkage operator ${\rm soft}_{\beta_k\mu}$ in the algorithm.
Each nonzero component of 
$\beta_k \mu \gamma^{-1}(x_k - {\rm soft}_\gamma(x_k))$
shares the same sign as 
the corresponding component of $x_k$.
Hence, this term debiases the estimate by enhancing the magnitudes of
the nonzero components prior to the operation of ${\rm soft}_{\beta_k\mu}$,
while maintaining zero components.

In the case of PMC,
the projection $P_{\mathcal{M}}$ restricts
the ``debiasing'' effect (the Moreau enhancement effect) to the subspace
$\mathcal{M}$.
Here, this restriction is due to a requirement for ensuring convexity of
the whole cost of \eqref{eq:sparse_regression_mc_underdetermined}.
We therefore refer to the algorithm in \eqref{eq:isda} 
as {\em the iterative shrinkage and debiasing algorithm
(ISDA)},\footnote{Although \eqref{eq:isda} can be regarded as a specific instance of the iterative
shrinkage-thresholding algorithm (ISTA) \cite{beck09},
we call it ISDA due to its remarkable debiasing property.
A stochastic version of ISDA
has been presented in \cite{kaneko20} with its geometric interpretation.
}
which converges to a minimizer of
\eqref{eq:sparse_regression_mc_underdetermined}
provided that the smooth part in
\eqref{eq:sparse_regression_mc_underdetermined_smooth_nonsmooth} is convex.
The convexity condition is given below.

\begin{proposition}[Convexity condition for
 \eqref{eq:sparse_regression_mc_underdetermined}]
\label{proposition:rmc_convexity}
The smooth part
$0.5\norm{\inputA x - \outputb}_2^2 -
\hspace*{.1em}^{\gamma}\norm{\cdot}_1(P_{\mathcal{M}}x)$ 
is convex
if and only if
$\mu \leq \gamma\lambda_{\min}^{++}(\inputA^\top \inputA)$,
where
$\lambda_{\min}^{++}(\cdot)$ denotes 
the smallest strictly-positive eigenvalue.

\end{proposition}
\begin{proof}
The proof is based on the results to be presented in Section
 \ref{subsec:convexity_condition}, and it is given in
Appendix \ref{subsec:proof_pmc}.
\end{proof}


\subsection{SORR Estimator for Outlier-Robust Regression}\label{subsec:robust_regression}

Robust regression concerns the case when some components of $\outputb$
contaminate outliers as follows \cite{candes_randall08}:
$\outputb:= A x_{\star} + \varepsilon_{\star} + o_{\diamond}$.
Here, $x_{\star}\in\real^n$ and $\varepsilon_{\star}\in\real^m$
are the unknown and noise vectors 
which are mutually uncorrelated and both of which
obey i.i.d.~zero-mean normal distributions with variances
$\sigma_{x_{\star}}^2\in\real_{++}$ and 
$\sigma_{\varepsilon_{\star}}^2\in\real_{++}$, respectively, 
and $o_{\diamond}\in\real^m$ is the sparse outlier vector
\cite{huber_book}.



\begin{figure}[t]
	\centering

\begin{tabular}{cc}
\begin{minipage}{4cm}
 \subfigure[loss function $\phi$]{
\hspace*{-2.2em}  
\includegraphics[height=3.5cm]{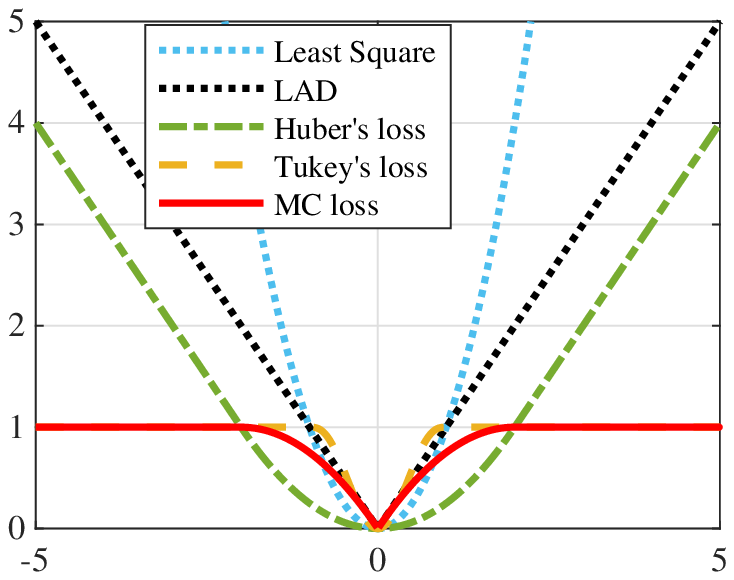}
 }\vspace*{-1em} 
\end{minipage}
 &
\begin{minipage}{4cm} 
\subfigure[derivative $\psi$]{
\hspace*{-1.8em} \includegraphics[height=3.5cm]{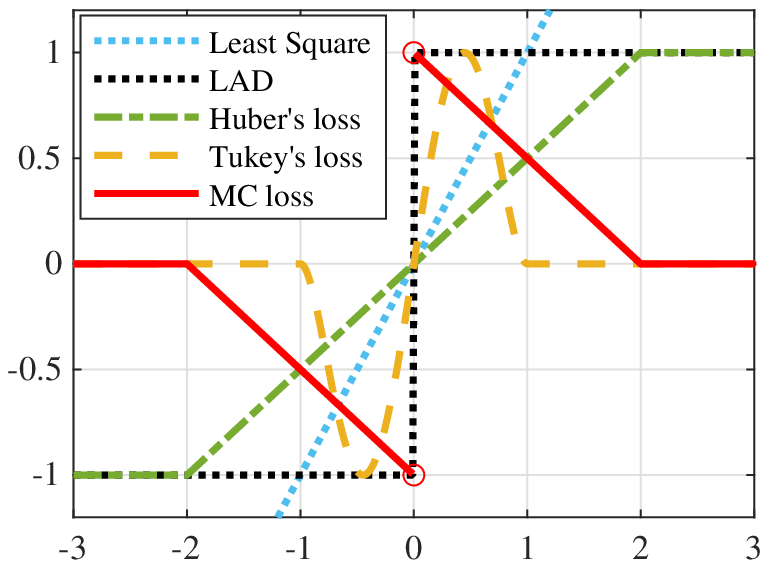}
}\vspace*{-1em}
\end{minipage}
 \\
\end{tabular}
	\caption{Loss functions and the derivatives (when exist).}
	\label{fig:MC_function}
\end{figure}

\subsubsection{A tradeoff between robustness and mathematical
 tractability, and stability aspect}
\label{subsubsec:sorr_motivation}

Popular Huber's loss \cite{huber_book} is more insensitive to
outliers than the least square (LS) loss, and it is mathematically tractable
owing to its convexity at the same time.
Regarding stability with respect to fluctuations caused by Gaussian noise,
the Huber's and LS losses are equally stable.
Nevertheless, the robustness of Huber's loss 
against huge outliers is limited.
This can be seen by inspecting its derivative to which
the so-called {\em influence function} is proportional \cite{hampel2011robust}.
See Fig.~\ref{fig:MC_function}.
It can be seen that the derivative of Huber's loss stays constant above
(or below) the threshold.
This means that 
large outliers give a constant amount of influence to the estimate, 
thus causing extra estimation bias.

In contrast, Tukey's biweight loss
\cite{beaton_tukey74} has a {\em ``redescending property (vanishing derivative)''}; i.e., 
the derivative vanishes at some point on each side of the real line.
This implies that such outliers that have magnitudes exceeding the threshold
would give no influence to the estimate, thus leading to remarkable
robustness to huge outliers.
Unfortunately, however, Tukey's biweight is mathematically intractable
owing to its nonconvexity.
So, {\em how can we break the tradeoff between robustness 
against huge outliers and mathematical tractability.}

To find an answer to this question,
let us consider the following question first:
{\em can we find such a convex loss that has
a vanishing derivative (for sufficiently large values)?}
This is hopeless actually in a certain sense.
To be precise,
we restrict ourselves to such a class of loss functions
$\phi:\real\rightarrow [0,+\infty)$ such that
(i) $\phi(0)=0$, 
(ii) $\phi(e)> 0$ if $e\neq 0$, and 
(iii) $\phi$ is differentiable everywhere but the origin.
(This assumption is reasonably mild. See, e.g., \cite{selesnick17_nonseparable_reg}.)
Within this class of functions,
$\phi$ is continuous if it is convex, because ${\rm range} \ (\phi) = [0,+\infty)\subset\real$ owing
to conditions (i) and (iii).
In fact, no convex loss has a vanishing derivative in this case.
The derivative $\psi:=\phi'$ of $\phi$ has the following properties:
(i) $\psi(0)=0$ if $\phi$ is also differentiable at $e=0$ 
which minimizes $\phi$
(or $0\in \partial \phi(0)$ in general), and
(ii) $\psi(e)>0$ when $e$ increases from zero
infinitesimally.
Hence, there is no way for $\psi$ to vanish again
because the derivative of a convex function is monotonically non-decreasing.

The above arguments encourage us to explore nonconvex loss functions.
In fact, the MC loss $\Phi_{\gamma}^{\rm MC}$ 
has a vanishing derivative (as can be seen from
Fig.~\ref{fig:MC_function}), 
and it is mathematically tractable at the same time.
Let us now highlight the behaviour of the derivative $\psi$ for each
loss in the vicinity of the origin.
It can be seen that the derivative vanishes at the origin for
the Huber, Tukey's biweight, and  LS losses, while
it does not vanish for the MC and LAD losses.
This implies a direct use of the MC loss may cause instability 
with respect to small fluctuations generated by Gaussian noise.
We therefore present another formulation using an additional variable
vector to model the Gaussian noise vector $\varepsilon_{\star}$
in the following.

\subsubsection{Stable outlier-robust regression}
\label{subsubsec:sorr}

We introduce the variable vector\footnote{One may try to
introduce, instead of $\varepsilon$,
an additional variable vector to model the outlier $o_{\diamond}$.
This, however, leads to a nonconvex formulation.}
 $\varepsilon\in\real^m$
to model the noise $\varepsilon_{\star}$.
The SORR formulation is given as follows:
\begin{align}
\hspace*{-1em} \min_{x\in\real^n,\varepsilon\in\real^m} 
~&\mu 
\underbrace{[\norm{\outputb - (\inputA x + \varepsilon)}_1
- 
  ~^\gamma\norm{\cdot}_1(\outputb - (\inputA x + \varepsilon))]}_{\Phi_{\gamma}^{\rm
  MC}(\outputb - (\inputA x + \varepsilon))}\nonumber\\
\hspace*{-1em}~&+ 0.5 \sigma_x^{-2}\norm{x}_2^2
+ 0.5 \sigma_{\varepsilon}^{-2}\norm{\varepsilon}_2^2,
\label{eq:stable_regression}
\end{align}
where
$\sigma_x^2\in\real_{++}$ and $\sigma_\varepsilon^2\in\real_{++}$ are
estimates
of $\sigma_{x_{\star}}^2$ and 
$\sigma_{\varepsilon_{\star}}^2$, respectively.
If such estimates are unavailable, 
$\sigma_x^2$ and $\sigma_\varepsilon^2$ are considered as tuning
parameters.
An extreme case of SORR with $\sigma_\varepsilon\downarrow 0$
(or with a sufficiently small $\sigma_\varepsilon>0$) 
make the optimal $\varepsilon$ of \eqref{eq:stable_regression} be 
the zero vector, reducing SORR to the following simple formulation:
\begin{equation}
 \min_{x\in\real^n} ~ 
\mu \underbrace{[\norm{\inputA x - \outputb}_1
- 
  ~^\gamma\norm{\cdot}_1(\inputA x - \outputb)]}_{\Phi_{\gamma}^{\rm MC}(\inputA x - \outputb)}
+ 0.5\norm{x}_2^2.
\label{eq:robust_estimation_mc_loss}
\end{equation}
We shall refer to the formulation in \eqref{eq:robust_estimation_mc_loss}
as outlier-robust regression (ORR).\footnote{
Unlike SORR, the ORR formulation in \eqref{eq:robust_estimation_mc_loss}
does not distinguish the Gaussian noise $\varepsilon_{\star}$ and the sparse outlier
$o_{\diamond}$ explicitly.
ORR in \eqref{eq:robust_estimation_mc_loss} can also be viewed
as a particular case of the model proposed in \cite{suzuki20}
for robust recovery of jointly sparse signals.}

The first term $\Phi_{\gamma}^{\rm MC}(\outputb - (\inputA x +
\varepsilon))$
of \eqref{eq:stable_regression}
is the MC loss encouraging sparsity of the estimation residual
$\outputb - (\inputA x + \varepsilon)$ which can be regarded as an
estimate of the sparse outlier.
The last two terms
$0.5\sigma_{x}^{-2}\norm{x}_2^2$ and
$0.5\sigma_{\varepsilon}^{-2}\norm{\varepsilon}_2^2$
reflect the Gaussianity of $x_{\star}$ and $\varepsilon_{\star}$,
playing double roles of convexification and regularization (in the Tikhonov sense).
In particular, when the noise power
$\sigma_{\varepsilon_{\star}}^2$ is large,
$\norm{\varepsilon_{\star}}^2$
tends to be large as well.
In this case,
the inverse $\sigma_{\varepsilon}^{-2}$ of 
the noise-power estimate would be small, and it
allows
$\norm{\varepsilon}_2^2$ to be large 
so that $\varepsilon$ mimics $\varepsilon_{\star}$ well,
yielding  efficient mitigation of the MC loss
$\Phi_{\gamma}^{\rm MC}(\outputb - (\inputA x +
\varepsilon))$.
This leads to the ``stability'' of the SORR estimator in the spirit of \cite{zhou10}.
A primal-dual splitting algorithm 
which can solve some class of linearly-involved nonsmooth convex
optimization problems including \eqref{eq:stable_regression} will be
presented in Section \ref{subsec:twin_prox}.
The algorithm relies on convexity of the smooth part of the objective function
in \eqref{eq:stable_regression}, for which the condition is given below.

\begin{proposition}[Convexity condition for SORR \eqref{eq:stable_regression}]
\label{proposition:orsr_convexity}
 The smooth part 
$0.5\sigma_{x}^{-2}\norm{x}_2^2 +
0.5 \sigma_{\varepsilon}^{-2}\norm{\varepsilon}_2^2
-  \mu\hspace*{.3em} ^\gamma\norm{\cdot}_1(\outputb - (\inputA x + \varepsilon))$ 
is convex in $(x,\varepsilon)\in\real^n\times\real^m$
if and only if
\begin{equation}
 \mu (\sigma_{\varepsilon}^2 + \sigma_x^2
  \lambda_{\max}(\inputA^\top\inputA))\leq \gamma.
\label{eq:mu_condition_stable_regression}
\end{equation}
\end{proposition}
\begin{proof}
The proof is based on the results to be presented in Section
 \ref{subsec:convexity_condition},
and it is given in Appendix \ref{subsec:proof_orsr}.
\end{proof}

\begin{remark}[SORR resolves the tradeoff efficiently]\label{remark:sorr_tradeoff}
The SORR estimator breaks the tradeoff 
between robustness and mathematical tractability.
In particular, SORR enjoys (i) remarkable robustness against huge
outliers and (ii) insensitivity to small fluctuations,
while the posed problem in \eqref{eq:stable_regression} is
still tractable because the whole cost is convex under \eqref{eq:mu_condition_stable_regression}.
Those advantages come mainly from the use of the MC loss and
the additional vector $\varepsilon$.
\end{remark}

We mention that the introduction of the additional vector $\varepsilon$
does not increase computational complexity essentially,
although a larger amount of memory is required than the case of ORR
(the $\varepsilon$-parameter free formulation)
to store the length-$(n+m)$ vector $\xi$
as well as some other intermediate vectors 
that are required in the algorithm 
(see 
\cite{yukawa_eusipco22} or Section \ref{subsubsec:pdda_type_R}).
Specifically, the algorithm iteration to compute the SORR estimator
requires $O(mnQ)$ complexity, in addition to the computation
of the largest eigenvalue of $A^\top A$ (or $A A^\top$) as preprocessing,
where $Q$ is the number of iterations.

The SORR formulation has been extended to sparse modeling under
Gaussian and impulsive noises \cite{suzuki23}.




\section{LiMES Model: Convexity Condition, Algorithm, and Applications}\label{sec:alime}

To give the convexity conditions for the debiased sparse modeling
and SORR in a unified way, 
we present a generalized model called ``LiMES''
and show the necessary and sufficient condition for its convexity.
Applications of the LiMES model can be classified into two categories:
type-sparse and type-robust.
We derive {\em the proximal debiasing-gradient algorithm} (which
requires no auxiliary variable) for the former type
and {\em the primal-dual debiasing algorithm} for the latter type.
We finally give a couple of other applications than debiased sparse modeling
and SORR.

Let $\genopL:\hilbertarbh\rightarrow \hilbertarbk$ be
a bounded linear operator 
from a Hilbert space $\hilbertarbh$ to another Hilbert space
$\hilbertarbk$.
The adjoint operator of $\genopL$ is denoted by $\genopL^*$.
The operator norm is then defined by
$\norm{\genopL} := \sup \{\norm{\genopL x} \mid
x\in\hilbertarbh, \norm{x}\leq 1\}$.
Given a bounded linear operator 
$\genopL:\hilbertarbh\rightarrow \hilbertarbh$,
$\genopL\succeq O$ means that
$\genopL$ is positive semidefinite, i.e.,
$\innerprod{\genopL x}{x}\geq 0$ for all $x\in\euclidspace$.
Given any bijective bounded linear operator
$\genopL:\euclidspace\rightarrow\euclidspace$
and any function $f\in\gzh$, 
it holds that
\begin{equation}
(f\circ \genopL)^*=f^*\circ (\genopL^*)^{-1}.
\label{eq:fL_conjugate} 
\end{equation}


\subsection{LiMES: A class of weakly convex functions}\label{subsec:alime}

\begin{definition}[The LiMES Model]
\label{def:limes}

Let $\hilbertx$, $\hilbertnewy$, and $\hilbertnewz$ be finite-dimensional
Hilbert spaces.
Let $\opA_1: \hilbertx\rightarrow  \hilbertnewy:x\mapsto
M_1 x + c_1$ and $\mu\in\real_{++}$,
where
$(O\neq) M_1:\hilbertx\rightarrow  \hilbertnewy$ is a bounded linear operator and 
$c_1 \in\hilbertnewy$ is a vector.
Let $(O\neq)$
$\opL:\hilbertnewz\rightarrow \hilbertnewz$ be a 
bounded linear operator\footnote{
The letter $\opL$ will be used to denote the linear operator of
LiMES, distinguished from the general linear operator $\genopL$
(which was used to denote the linear operator of LiGME in
\cite{abe_ip20}).},
$\opD:\hilbertnewz\rightarrow \hilbertnewz$
be a diagonal positive-definite operator,
and
$\opA_2:\hilbertx \rightarrow \hilbertnewz :x\mapsto \opM x + \affinec$,
where
$(O\neq)\opM:$
$\hilbertx \rightarrow \hilbertnewz $ is a bounded linear operator
and $\affinec\in\hilbertnewz$ is a vector.
Let $\Psi\in\Gamma_0(\hilbertnewz)$, which is referred to as a seed function.
The {\em linearly-involved Moreau-enhanced-over-subspace}
(LiMES) model is defined as the minimization of
the following function:
\begin{equation}
J_{\Omega}:
\hilbertx\rightarrow (-\infty,\infty]:
x\mapsto 
0.5\norm{\opA_1 x}^2 +
\mu\Psi_\opD^\opL (\opA_2 x),
\label{eq:limes_model}
\end{equation}
where $\Omega:=(\opA_1; \Psi_D^{\opL}\circ \opA_2)$, and
\begin{align}
\hspace*{-1.7em} \Psi_D^{\opL} :&~ \hilbertnewz\rightarrow 
 (-\infty,+\infty] 
\nonumber\\
\hspace*{-3em}:&~ \newz\mapsto
\Psi(\newz) - \min_{\newv\in\hilbertnewz} \big[\Psi(\newv) +
0.5\norm{\opD(\opL\newz - \newv)}^2 \big].
\label{eq:mes}
\end{align}
We refer to 
$\Psi_\opD^\opL \circ \opA_2:\hilbertx \rightarrow (-\infty,+\infty]$
as {\it the LiMES function}.
\end{definition}

Define the subspace $\mathcal{M}_1:=\range M_1^*$.
The debiased sparse modeling in
\eqref{eq:sparse_regression_mc_underdetermined}
is reproduced 
by letting
$\hilbertx:=\hilbertnewz:=\real^n$, $\hilbertnewy:=\real^m$,
 $\opA_1:=A\cdot - \outputb$ ($M_1:=A$), 
$\Psi:=\norm{\cdot}_1$,
$\opA_2:=I_n$ ($M_2:= I_n$), $\opL:=P_{\mathcal{M}_1}=P_{\mathcal{M}}=A^{\dagger} A\in\real^{n\times
n}$, and  $D:=\gamma^{-1/2}I_n$.
On the other hand, ORR in  \eqref{eq:robust_estimation_mc_loss}
is reproduced by letting
$\hilbertx:=\hilbertnewy:=\real^n$, $\hilbertnewz:=\real^m$,
$\opA_1:=I_n$,
$\Psi:=\norm{\cdot}_1$,
$\opA_2:=A\cdot - \outputb$, 
$\opL:=I_m$, and  $D:=\gamma^{-1/2}I_m$.
In the former example, here, 
the quadratic term 
$0.5\norm{\opA_1 x}^2$ of \eqref{eq:limes_model}
represents the data fidelity, and 
the second term $\mu\Psi_\opD^\opL (\opA_2 x)$ represents the penalty.
Hereafter, we shall refer to this type as {\em type-sparse}, or {\em type-S}
for short.
In the latter example, on the other hand, 
the roles of the two terms are reversed, and
we refer to this type as {\em type-robust}, or {\em type-R}.
As will be seen in Section \ref{subsec:robust_stable_LiMES}, SORR
in \eqref{eq:stable_regression} is also a particular example of the
LiMES model.
Typical roles of the linear/affine operators are summarized in 
Table \ref{table:roles_of_operators}.


\begin{table}[t!]
\caption{Typical roles of linear/affine operators}
\label{table:roles_of_operators} 
\centering
\begin{tabular}{|l|l|}
\hline
\!$\opL$ \!\!\!& preserving the convexity of the smooth term (default:
     $P_{\range M_1^*}$)\\ \hline
\!$D$\! \!\!& assigning individual weights to the variables (Section \ref{subsec:spcp})\\ \hline
\!$\opA_1$\!\!\! & used to define data fidelity for type-S applications \\ \hline
\!$\opA_2$\!\!\! & used to define data fidelity for type-R applications \\ \hline
\end{tabular}
\end{table}

\begin{table}[t!]
\caption{Mathematical notation}
\label{table:roles_of_operators} 
\centering
\begin{tabular}{|l|l|}
\hline
 \!$\iota_C$ \!\!\!& the indicator function of $C$ \\ \hline
\!$\sigma_C$ \!\!\!&  the support function of $C$\\ \hline
\!$f^*$ \!\!\!& the Fenchel conjugate of $f$ \\ \hline
\!$\norm{\cdot}_*$ \!\!\!& the dual norm of $\norm{\cdot}$\\ \hline
\!$^\gamma f$ \!\!\!& the Moreau envelope of $f$, e.g.,
$^\gamma \norm{\cdot}_1$ \\ \hline
\!${\rm Prox}_f$ \!\!\!& the proximity operator w.r.t.~$f$\\ \hline
\!$P_{\mathcal{M}}$ \!\!\!& the projection operator onto subspace $\mathcal{M}$\\ \hline
\!$\genopL^*$ \!\!\!& the adjoint of linear operator
     $\genopL$\\ \hline
\!$\Psi_D^{\opL}$ \!\!\!& the LiMES function ($\Psi_D^{I}=\Psi_D$): see \eqref{eq:mes} \\ \hline
\!$J_{\Omega}$ \!\!\!& the LiMES model for $\Omega:=(\opA_1;
     \Psi_D^{\opL}\circ \opA_2)$: see \eqref{eq:limes_model}
 \\ \hline
\end{tabular}
\end{table}


We now discuss an issue related to
the overall convexity of the function
$J_{\Omega}$
in \eqref{eq:limes_model}.
Due to the nonsingularity of $D$,
it can be verified that
\begin{align}
\hspace*{-1.7em} \Psi_D^{\opL} (\newz) 
&=\Psi(\newz) - 
\min_{\tilde{\newv}\in\hilbertnewz} \big[\Psi(D^{-1}\tilde{\newv}) +
0.5\norm{\opD\opL\newz -\tilde{\newv}}^2 \big]
\nonumber\\
&=\Psi(\newz) - \hspace*{.1em} ^1(\Psi\circ
 \opD^{-1})(\opD\opL\newz)  \label{eq:phiGL_moreau1}\\
\hspace*{-3em}&=\Psi(\newz) - 0.5\norm{\opD\opL\newz}^2
+ \hspace*{.1em} ^1(\Psi^*\circ
 \opD)(\opD\opL\newz),  \label{eq:phiGL_moreau1b}
\end{align}
where the last equality is verified by 
\eqref{eq:moreau_identity} and \eqref{eq:fL_conjugate}
together with
the self-adjointness $D^*=D$.\footnote{
The usefulness of the identity given in \eqref{eq:moreau_identity} 
in considering overall convexity
has
been witnessed already in the contexts of graph learning \cite{koyakumaru21,koyakumaru22}
and distributed optimization \cite{komuro22}.}
Here, the first and third terms of \eqref{eq:phiGL_moreau1b} are convex
functions of $z$.
Since convexity is preserved under composition
with an affine operator \cite[Proposition 8.20]{combettes},
the LiMES function $\Psi_D^{\opL} \circ \opA_2$ is
$\eta$-weakly convex if
$ 0.5\eta\norm{\cdot}^2 - 0.5\norm{\cdot}^2 \circ\opD\opL\opA_2$ is convex for some $\eta\in\real_{++}$,
or equivalently if 
$\eta I - M_2^*\opL^* \opD^2\opL M_2\succeq O~
(\Leftrightarrow \eta\geq \norm{\opD\opL M_2}^2)$.
Substituting \eqref{eq:phiGL_moreau1} into \eqref{eq:limes_model}
yields the following smooth-nonsmooth separation:
\begin{equation}
\hspace*{-.7em}
J_{\Omega}
\!=\! \underbrace{0.5\norm{\cdot}^2\circ \opA_1 \!-\! \mu ~^1(\Psi\circ \opD^{-1}) \!\circ\! \opD\opL\opA_2}_{=: F~{\rm (smooth)}} \!+  \!\!
\underbrace{\mu \Psi \!\circ\! \opA_2.}_{{\rm nonsmooth}}
\label{eq:fg_smooth_nonsmooth}
\end{equation}
Because our algorithms to be presented in Section \ref{subsec:twin_prox}
treat the smooth and nonsmooth terms separately,
both of those terms need to be convex 
for ensuring convergence to a global minimizer.
The convexity condition for the smooth part $F$ will be discussed
in Section \ref{subsec:convexity_condition},
as the nonsmooth term is automatically convex due to the convexity of $\Psi$.

For consistent notation with \cite{abe_ip20},
$\Psi_\opD := \Psi_\opD^I$ will be used when $\opL:=I$.
The question now is: {\it what is the role of 
the term
$\min_{\newv\in\hilbertnewz} \big[\Psi(\newv) +
0.5\norm{\opD(\opL\newz - \newv)}^2 \big]$ in
\eqref{eq:mes}?}
The following proposition, which generalizes Proposition
\ref{proposition:rmc_properties}, answers this question
for the case of $\opL := P_{\mathcal{M}_1}$.

\begin{proposition}
\label{proposition:limes_properties}

\begin{enumerate}
  \item [(a)]
The particular LiMES function $\Phi_{\gamma}^{P_{\mathcal{M}_1}}$
	coincides with the generalized Moreau enhanced penalty $\Psi_D$
{\rm \cite{abe_ip20}}
on the subspace $\mathcal{M}_1$; i.e.,
$\Psi_D^{P_{\mathcal{M}_1}}(z)
 =  \Psi_D^{I}(z)
=  \Psi_D(z)
= \Psi(z) 
- \min_{\newv\in\hilbertnewz} \big[\Psi(\newv) +
0.5\norm{\opD(\newz - \newv)}^2
 \big]$ for 
$z\in \mathcal{M}_1$.

 \item [(b)] Let $\opL$ satisfy $\opL = \opL \circ P_{\mathcal{M}_1}$.
Then, 
$\Phi_{\gamma}^{\opL}$
	reduces to $\Psi$ (up to constant)
on $\mathcal{M}_1^\perp$; i.e.,
$\Psi_D^{\opL}(z) = \Psi(z)
- \underbrace{\min_{\newv\in\hilbertnewz} \big[\Psi(\newv) +
0.5\norm{\opD\newv}^2 \big]}_{{\rm constant ~in ~} z}$ for 
$z\in \mathcal{M}_1^{\perp}$.
%

\end{enumerate}
\end{proposition}
\begin{proof}
(a) The assertion can be verified by applying $P_{\mathcal{M}_1} z=z$
 for all $z\in \mathcal{M}_1$
to \eqref{eq:phiGL_moreau1} with $\opL:=P_{\mathcal{M}_1}$.\\
(b) Use $\opL z =\opL \circ P_{\mathcal{M}_1}z =\opL 0=0$,
 $\forall z\in \mathcal{M}_1^{\perp}$ in \eqref{eq:mes}.
\end{proof}

Proposition \ref{proposition:limes_properties} states,
under the use of $\opL:=P_{\mathcal{M}_1}$,
that $\Psi_D^{P_{\mathcal{M}_1}}$ is an ``exact'' Moreau-enhanced model
over the subspace $\mathcal{M}_1$
in the sense of {\em the generalized Moreau enhanced (GME) penalty}
\cite{abe_ip20}.
Note here that $\min_{\newv\in\hilbertnewz} \big[\Psi(\newv) +
0.5\norm{\opD(\newz - \newv)}^2
 \big]$ can be regarded as
{\em a generalized Moreau envelope} of $\Psi$.
In all applications presented in this article,
$\opL:=P_{\mathcal{M}_1}$ will be used.
Nevertheless, we would not exclude the possibility of using 
other choices of $\opL$ such as those presented in \cite{lanza19,abe_ip20},
although the Moreau enhancement over $\mathcal{M}_1$ 
could be ``inexact'' in this case
(see Example \ref{example:alime_penalty} in Section \ref{subsec:examples}).
Proposition \ref{proposition:limes_properties}(b) states that 
$\Psi_D^{\opL}$ coincides with $\Psi$ over $\mathcal{M}_1^\perp$ up to constant under the
condition (which $\opL:=P_{\mathcal{M}_1}$ satisfies).

As shown in Section \ref{subsec:sparse_regression},
the PMC penalty preserves the convexity of the smooth part 
even when $A^\top A$ is singular,
and at the same time it enjoys the mixed nature of separability and nonseparability.
Such a penalty can be generated systematically
by the LiMES function with $\opL:=P_{\mathcal{M}_1}$
given a separable function $\Psi$.
We emphasize here that the diagonality of $D$ induces
the separability of the function
$\Psi(\newv) +
0.5\norm{\opD(\opL\newz - \newv)}^2$ in \eqref{eq:mes}
in terms of $v$, which makes the gradient computation
of the smooth part $F$ in \eqref{eq:fg_smooth_nonsmooth} simple
(see Remark \ref{remark:separability_of_rmc}).
In particular,
for typical type-S applications
(such as debiased sparse modeling and SPCP to be presented in Section \ref{subsec:spcp}),
$\opA_2$ is also a diagonal operator,
and thus 
the computationally efficient proximal gradient algorithm
can be applied which requires no auxiliary vector
to compute the LiMES model
(see Section \ref{subsec:twin_prox}).

To show an active role of the diagonal operator $D$ briefly,
suppose that the variable vector $x\in\hilbertx$ consists of several
subvectors.
In this case, $\opD$ can be used to give an individual weight to the regularizer
of each subvector (see Section \ref{subsec:spcp}).

\subsection{Examples of LiMES function: penalty and loss}
\label{subsec:examples}

In this subsection, we simply let $\opD:=\gamma^{-1/2}I$ for $\gamma\in\real_{++}$,
which reduces \eqref{eq:phiGL_moreau1} to
\begin{equation}
\hspace*{0em} 
\Psi_D^{\opL}(z) =
\Psi_{\gamma^{-1/2}I}^{\opL} (z)
= \Psi(z) - \hspace*{.1em} ^{\gamma}\Psi(\opL z).  \label{eq:phiGL_moreau2}
\end{equation}
Some examples of the LiMES function are listed below.

\begin{example}[LiMES penalty]
\label{example:alime_penalty}
We let $\hilbertx:=\real^n$ and $\hilbertnewz:=\real^m$ in (a) -- (d) below.
 \begin{enumerate}
  \item [(a)] (MC penalty {\rm \cite{zhang,selesnick}})
Let $\Psi:=\norm{\cdot}_1$, $\opL:=\opA_2:=I_n$ ($n=m$).
Then, $(\norm{\cdot}_1)_{\gamma^{-1/2}I}
:= \norm{\cdot}_1 - \hspace*{0em}^{\gamma}\norm{\cdot}_1 = \Phi_{\gamma}^{\rm MC}$.
In particular, the MC penalty, or 
$\theta_{\gamma} (\norm{\cdot}_1)_{\gamma^{-1/2}I}$ more specifically,
gives a parametric bridge between $\norm{\cdot}_0$ and $\norm{\cdot}_1$
{\rm \cite{abe_ip20}} (see Remark \ref{remark:bridge} for definition of $\theta_{\gamma}$).



  \item [(b)] { (PME and PMC)}
Let $\opL:=P_{\mathcal{M}_1}$ and $\opA_2:=I_n$ ($n=m$), where
$\mathcal{M}_1\subset\real^n$ is a linear subspace of $\real^n$.
Then, $\Psi_{\gamma^{-1/2}I}^{P_{\mathcal{M}_1}}
= \Psi - \hspace*{.1em}^{\gamma}\Psi ~\circ~ P_{\mathcal{M}_1}$,
which we call {\em the projective Moreau enhanced (PME) function}.
In particular, letting $\Psi:=\norm{\cdot}_1$ yields
$(\norm{\cdot}_1)_{\gamma^{-1/2}I}^{P_{\mathcal{M}_1}}
:= \norm{\cdot}_1 - \hspace*{0em}^{\gamma}\norm{\cdot}_1 \circ
	P_{\mathcal{M}_1}
= \Phi_{\gamma}^{\rm PMC}$,
which is the PMC penalty presented in 
Section	\ref{subsec:sparse_regression}.
An alternative choice of $\opL$ to the $P_{\mathcal{M}_1}$ used in 
$\Phi_{\gamma}^{\rm PMC}$
is given by
$\opL:= \sqrt{\gamma/\mu} V 
{\rm diag}(\alpha_1^{1/2},\alpha_2^{1/2},\cdots,\alpha_n^{1/2})
\Sigma_{A^\top A}^{1/2} V^\top$ (cf.~{\rm \cite{lanza19}}),
where $\alpha_i\in[0,1]$, $i=1,2,\cdots,n$, 
are tuning parameters,
and $A^\top A = V \Sigma_{A^\top A} V^\top$
is an eigenvalue decomposition
with some orthogonal matrix $V\in\real^{n\times n}$
and some diagonal matrix $\Sigma_{A^\top A}\succeq O$.
(This choice actually satisfies the convexity condition to be presented
	in Section \ref{subsec:convexity_condition}.)

  \item [(c)] { (MC-W)}
Let $\Psi:=\norm{\cdot}_1$, $\opL:=I_n$ ($n=m$), and $\opA_2:=\mathscr{W}$,
where $\mathscr{W}$ is the popular wavelet transform {\rm \cite{mallat_book}}.
Then, $(\norm{\cdot}_1)_{\gamma^{-1/2}I} \circ \mathscr{W}$
is {\em the MC wavelet (MC-W)}.

  \item [(d)] (MC-TV)
Let $\Psi:=\norm{\cdot}_1$, $\opL:=I_{n-1}$
($m=n-1$), and
 $\opA_2:=\mathscr{D}_n:= [0_{n-1}~I_{n-1}] -[I_{n-1} ~ 0_{n-1}] \in\real^{(n-1)\times n}$
 be the first-order differential operator,
where $0_n:=[0,0,\cdots,0]^\top\in\real^n$ for any $n\in\Natural^*$.
Then, $(\norm{\cdot}_1)_{\gamma^{-1/2}I} \circ \mathscr{D}_n$
has been used in 
{\em MC total-variation (MC-TV) denoising} {\rm \cite{du18}}.

  \item [(e)] { (MEN)}
Let $\hilbertx:=\hilbertnewz:=\real^{n\times m}$, and
$\Psi:=\norm{\cdot}_{\rm nuc}$, which is the nuclear norm (the sum of the singular
	values) of a matrix,
and $\opL:=\opA_2:=I$.
Then, $ (\norm{\cdot}_{\rm nuc})_{\gamma^{-1/2}I}$ gives
{\em the Moreau enhanced nuclear-norm (MEN)}.
In particular, the normalized version $2\gamma^{-1}(\norm{\cdot}_{\rm nuc})_{\gamma^{-1/2}I}$
gives a parametric bridge between the rank of matrix and
	$\norm{\cdot}_{\rm nuc}$
{\rm \cite{abe_ip20}}.
The MEN penalty will be used in Section \ref{subsec:spcp}
for SPCP.

 \end{enumerate}
\end{example}


\begin{example}[LiMES loss]
\label{example:alime_loss}
We let $\hilbertx:=\real^n$ and $\hilbertnewz:=\real^m$
in (a) -- (c) below, and
$\inputA\in\real^{m\times n}$ and $\outputb\in\real^m$
in (a), (c), and (d).

 \begin{enumerate}

  \item [(a)] { (MC loss)}
Let $\Psi:=\norm{\cdot}_1$, $\opL:=I_m$, and
$\opA_2:\real^n\rightarrow\real^m:x\mapsto \inputA x -\outputb$.
Then, $(\norm{\cdot}_1)_{\gamma^{-1/2}I} (\inputA \cdot-\outputb) = 
\Phi_{\gamma}^{\rm MC}(\inputA \cdot-\outputb)$ gives an MC loss, which has been
studied in Section \ref{subsec:robust_regression} for robust
	regression.


  \item [(b)] { (ME-hinge loss)}
Let 
$\Psi:=\sigma_{[-1,0]}:\real\rightarrow\real:z\mapsto
	\sup_{v\in[-1,0]} vz= \max\{0,-z\}$,
$\opL:=I_m=1$ ($m:=1$), and
$\opA_2:\real^n\rightarrow\real:x\mapsto a^\top x - 1$
for some given $a\in\real^n$ such that $\norm{a}_2=1$.
Then, 
$\Psi_{\rm hinge}:=\Psi\circ\opA_2:\hilbertx \rightarrow \real:
x\rightarrow \max\{0,1-a^\top x\}$
is the hinge loss function, and 
we call $ (\Psi_{\rm hinge})_{\gamma^{-1/2}I}
= \Psi_{\gamma^{-1/2}I}\circ\opA_2$
{\em the Moreau-enhanced hinge (ME-hinge) loss function}.
See Proposition \ref{proposition:me_hinge_equivalence} for the second equality here.
The proximity operator of $\Psi_{\rm hinge}$ is given for instance in
	{\rm \cite[Example 24.37]{combettes}}.
The ME-hinge loss will be used in Section \ref{subsec:robust_classification} for
	robust classification.

  \item [(c)] { (MC-W loss)}
Let $\Psi:=\norm{\cdot}_1$, $\opL:=I_m$, and 
$\opA_2:\real^{n}\rightarrow\real^m:x\mapsto \mathscr{W}(\inputA x - \outputb)$.
Then, $ (\norm{\cdot}_1)_{\gamma^{-1/2}I} (\mathscr{W}(\inputA\cdot - \outputb))$
gives an MC-W loss.

  \item [(d)] { (MC-TV loss)}
Let $\hilbertx:=\real^n$, $\hilbertnewz:=\real^{m-1}$,
$\Psi:=\norm{\cdot}_1$, $\opL:=I_{m-1}$, and 
$\opA_2:\real^{n}\rightarrow\real^{m-1}:x\mapsto \mathscr{D}_m(\inputA x -\outputb)$.
Then, 
$ (\norm{\cdot}_1)_{\gamma^{-1/2}I} (\mathscr{D}_m(\inputA\cdot -\outputb))$
gives an MC-TV loss.

\item [(e)] { (MEN loss)}
Let $\hilbertx:=\hilbertnewz:=\real^{n\times m}$,
$\Psi:=\norm{\cdot}_{\rm nuc}$, 
$\opL:=I$, and
$\opA_2:\real^{n\times m}\rightarrow\real^{n\times m}:X\mapsto X - Y$ 
given $Y\in\real^{n\times m}$.
Then, $(\norm{\cdot}_{\rm nuc})_{\gamma^{-1/2}I} \circ (\cdot - Y)$ 
gives a MEN loss.

 \end{enumerate}
\end{example}

\subsection{Convexity condition for the smooth part of  \eqref{eq:fg_smooth_nonsmooth}}
\label{subsec:convexity_condition}

We discuss the condition for convexity of the smooth part $F$,
which immediately implies the overall convexity of
$J_{\Omega}$
since the nonsmooth term $\mu\Psi\circ \opA_2$ is clearly convex.
By \eqref{eq:phiGL_moreau1} and \eqref{eq:phiGL_moreau1b},
the smooth part of \eqref{eq:fg_smooth_nonsmooth} can be rewritten
as
\begin{align}
\hspace*{-1.5em} F = 0.5 (\norm{\opA_1\cdot}^2
- \mu\norm{\opD\opL\opA_2\cdot}^2)
+
\mu \hspace*{.3em} ^1(\Psi^*\circ \opD) \circ \opD\opL\opA_2.
\label{eq:F}
\end{align}
Since the third term here is automatically convex,
$F$ is convex if the sum of the first two terms is convex; i.e., 
$F$ is convex if
\begin{equation*}
(\spadesuit)~~~ M_1^*M_1 - \mu \opM^*\opL^* D^2 \opL \opM  \succeq O.
\end{equation*}
In general, 
($\spadesuit$) is not a necessary condition.
When the third term of \eqref{eq:F} is strongly convex,
for instance,
$F$ could be convex even if 
$0.5\norm{\opA_1\cdot}^2
- 0.5\mu\norm{\opD\opL\opA_2\cdot}^2$
is nonconvex.
It can actually be observed that
the function $^1(\Psi^*\circ \opD)$ is strongly convex if and only if 
$\Psi$ is smooth (i.e., Fr\'echet differentiable with Lipschitz-continuous gradient)
due to {\rm \cite[Theorem 18.15]{combettes}}
together with 
$(^1(\Psi^*\circ \opD))^*= \Psi^{**} \circ \opD^{-1} +
0.5\norm{\cdot}^2
 = \Psi  \circ \opD^{-1} +  0.5\norm{\cdot}^2$
{\rm \cite[Proposition 13.24]{combettes}}.
Typical seed functions $\Psi$ including those presented in Section
\ref{subsec:examples} are nonsmooth, and
the above observation indicates that 
the third term of \eqref{eq:F} is not strongly convex
for such nonsmooth $\Psi$s.
In fact, ($\spadesuit$) is a necessary and sufficient condition
in those cases as well as many other cases.

To present the formal result regarding the convexity condition for $F$, 
we define the support function
$\Gamma_0(\hilbertnewz)\ni \sigma_C:z\mapsto \sup_{v\in C}\innerprod{z}{v}$ of
a nonempty closed convex set $C\subset \hilbertnewz$,
which is the conjugate function of the indicator function
$\Gamma_0(\hilbertnewz)\ni \iota_C:z\mapsto \left\{
\begin{array}{cc}
0 & \mbox{ if } z\in C \\
+\infty & \mbox{ if } z\not\in C,
\end{array}
\right.$
and hence $\sigma_C^*=\iota_C^{**}=\iota_C$.
Given an arbitrary norm $|||\cdot |||$
defined on the vector space $\hilbertnewz$,
the support function
$||| \cdot |||_*:=\sigma_C$  of its level set
$C:={\rm lev}_{\leq 1} |||\cdot |||$
is the {\em dual norm} of $|||\cdot |||$
{\rm \cite{horn_johnson13,boyd04_convexbook}}.
It is known that the dual of the dual norm is the original norm, i.e.,
$|||\cdot|||_{**}= |||\cdot|||$.\footnote{This is not true in general
in infinite dimensional vector spaces.}
An arbitrary norm defined on $\hilbertnewz$ can therefore be
represented as the support function of the level set of its dual norm. 

Given any bounded linear operator 
$\genopL:\hilbertarbh\rightarrow \hilbertarbk$
from a Hilbert space $\hilbertarbh$ to another Hilbert space
$\hilbertarbk$
and any subsets $C_{\hilbertarbh}\subset \hilbertarbh$
and $C_{\hilbertarbk}\subset \hilbertarbk$,
we define
$\genopL(C_{\hilbertarbh}):=\{\genopL x \mid x\in C_{\hilbertarbh}\} \subset \hilbertarbk$
and 
$\genopL^{-1}(C_{\hilbertarbk}):=\{x\in \hilbertarbh \mid \genopL x\in C_{\hilbertarbk}\}$.

\begin{proposition}[Convexity \hspace*{-.2em}condition
 \hspace*{-.1em}for \hspace*{-.1em}smooth \hspace*{-.1em}part
 \hspace*{-.1em}of \hspace*{-.1em}\eqref{eq:fg_smooth_nonsmooth}]
\label{proposition:positive_definite_case_necessary_condition}
~\vspace*{-1.2em}
\begin{enumerate}
 \item [(a)] $F \in \Gamma_0(\hilbertx)$ if 
condition ($\spadesuit$) is satisfied.

 \item [(b)] 
Let $\Psi:=\sigma_C$
with a nonempty closed convex set $C\subset \hilbertnewz$.
Then, the following statements hold.
\begin{enumerate}
 \item [(i)]
 Given any $x\in\hilbertx$, the following equivalence holds:
\begin{align}
\hspace*{-2em}& ~F(x) = 
0.5\norm{\opA_1 x}^2
- 0.5\mu\norm{\opD\opL\opA_2 x}^2
\nonumber\\
\hspace*{-2em}\Leftrightarrow &~ \hspace*{.1em}
^1(\sigma_C^*\circ \opD) (\opD\opL\opA_2 x) = 0 \nonumber\\
\hspace*{-1em}\Leftrightarrow  &~  x\in K_C := \{\hat{x}\in\hilbertx\mid D^2 \opL\opA_2
       \hat{x} \in C\}.
\end{align}

 \item [(ii)] Assume that 
\begin{equation}
 \interior K_C \neq \emptyset,
\end{equation}
where $\interior K_C $ is the interior of $K_C$.
Then, $F \in \Gamma_0(\hilbertx)$ if and only if
($\spadesuit$) is satisfied.
\end{enumerate}

\end{enumerate}

\end{proposition}
\begin{proof}
\noindent (a) 
It is clear under ($\spadesuit$) that
$
0.5\norm{\opA_1 \cdot}^2
- 0.5\mu\norm{\opD\opL\opA_2 \cdot}^2\in\Gamma_0(\hilbertx)$.
It can also be verified that
$\Psi\in\Gamma_0(\hilbertx) \Rightarrow 
\Psi^*\in\Gamma_0(\hilbertx) \Rightarrow 
\hspace*{0em}^1(\Psi^*\circ \opD) \circ
 \opD\opL\opA_2\in\Gamma_0(\hilbertx)$.

\noindent (b.i) For $v\in\hilbertnewz$, it can be verified that
\begin{align*}
\hspace*{-1.5em} ^1(\sigma_C^*\circ \opD)(v)
 = &~ \min_{\newz\in\hilbertnewz}~ 
\left[\iota_C(\opD \newz) + 
0.5\norm{v- \newz}^2
\right] \nonumber\\
 = &~ \!\!\min_{\newz\in \opD^{-1} (C)} \! 
0.5\norm{v- \newz}^2
 =: 0.5 d^2(v, \opD^{-1} (C)).
\end{align*}
It follows thus that
\hspace*{.1em}$^1(\sigma_C^*\circ \opD) \circ \opD\opL\opA_2
= 0.5 d^2(\opD\opL\opA_2 \cdot, \opD^{-1} (C))$, and
using this equality in \eqref{eq:F}
verifies that
$F(x) = 
0.5\norm{\opA_1 x}^2
- 0.5\mu\norm{\opD\opL\opA_2 x}^2\Leftrightarrow\hspace*{.1em}
^1(\sigma_C^*\circ \opD) (\opD\opL\opA_2 x) = 0 
\Leftrightarrow D \opL\opA_2 x \in D^{-1}(C)
\Leftrightarrow D^2 \opL\opA_2 x \in C
\Leftrightarrow x\in K_C$.

\noindent (b.ii)
Since the third term of \eqref{eq:F} vanishes over $\interior K_C \neq
\emptyset$ by 
Proposition
\ref{proposition:positive_definite_case_necessary_condition}(b.i), 
$F$ is nonconvex if condition ($\spadesuit$) is unsatisfied.
This implies the necessity of ($\spadesuit$).
The sufficiency is verified already in 
Proposition
\ref{proposition:positive_definite_case_necessary_condition}(a).
\end{proof}

Proposition \ref{proposition:positive_definite_case_necessary_condition}
shows the situation under which ($\spadesuit$) is 
a necessary and sufficient condition.
The necessity implies that the condition cannot be weaker, or, 
in other words, the parameter $\mu$ cannot exceed the upper
bound obtained from ($\spadesuit$).
We remark that
the diagonality and positive definiteness
imposed implicitly on $D$ in Proposition
\ref{proposition:positive_definite_case_necessary_condition}
can be relaxed straightforwardly by solely imposing bijectivity.

\begin{lemma}
\label{lemma:implication_open_mapping_theorem}
Let $\genopL:\hilbertx\rightarrow \hilbertnewz$ be a
bounded linear operator.
Given a nonempty set $(\emptyset\neq )C\subset \hilbertnewz$ and a point $\hat{x}\in\hilbertx$,
it holds that $\genopL\hat{x}\in \interior C$ implies $\hat{x}\in
 \interior \genopL^{-1}(C)$.
If $\genopL$ is surjective, 
$\genopL\hat{x} \in \interior C \Leftrightarrow \hat{x}\in \interior
 \genopL^{-1}(C)$.
\end{lemma}
\begin{proof}
We denote by
$\ball(x,\epsilon):= \{u\in\hilbertx \mid
\norm{u-x}<\epsilon \}$
an open ball centered at $x\in\hilbertx$ with radius
$\epsilon\in\real_{++}$.
Assume that $\genopL\hat{x}\in\interior C$.
Then, there exists some $\epsilon\in\real_{++}$ such that
$\ball(\genopL\hat{x},\epsilon)\subset C$.
It can then be shown straightforwardly that
$\ball(\hat{x},\epsilon/\norm{\genopL})\subset \genopL^{-1} (C)$,
and hence $\hat{x}\in\interior \genopL^{-1} (C)$.
The converse implication in the equivalence part
is an implication of the well-known open mapping theorem
 \cite{kreyszig78}.\footnote{
The open mapping theorem states that, if a bounded linear operator
$\genopL:\hilbertx\rightarrow \hilbertnewz$ is surjective, it maps an open
 set in $\hilbertx$ to an open set in $\hilbertnewz$.}
To see this, assume that $\genopL$ is surjective and 
that $\hat{x}\in\interior \genopL^{-1} (C)$.
Then, there exists an $\epsilon\in\real_{++}$ such that
$\ball(\hat{x},\epsilon)\subset \genopL^{-1} (C)$, and
the image $\genopL(\ball(\hat{x},\epsilon))$ is an open set
due to the open mapping theorem.
The inclusion $\genopL\hat{x} \in
 \genopL(\ball(\hat{x},\epsilon))\subset C$ due to definition of
 inverse mapping
thus implies $\genopL\hat{x}\in\interior
 \genopL (\ball(\hat{x},\epsilon)) \subset \interior C$.
\end{proof}

The following lemma gives a way of checking 
the nonemptiness condition of $\interior K_C$ for necessity in
Proposition
\ref{proposition:positive_definite_case_necessary_condition}.

\begin{lemma}
\label{lemma:intKC_condition} 
Consider the following statements:
(i) $\interior K_C \neq \emptyset$,
(ii) $\interior C \neq \emptyset$, and
(iii) $ D^2 \opL \opA_2 \hat{x}\in\interior C \neq \emptyset$
for some $\hat{x}\in\hilbertx$.
Then, 
(iii) $\Rightarrow$ (i).
If $\range (\opL \opM) =\hilbertnewz$,
(i) $\Leftrightarrow$ (ii).
\end{lemma}
\begin{proof}
By Lemma \ref{lemma:implication_open_mapping_theorem},
(iii) $\Rightarrow$
$\exists\hat{x}\in\hilbertx$,
$D^2 \opL \opM\hat{x}\in \interior (C-D^2\opL \affinec)$
$\Rightarrow$
$\exists\hat{x}\in\hilbertx$,
$\hat{x}\in \interior (D^2 \opL \opM)^{-1}(C-D^2\opL \affinec)=\interior K_C$
$\Rightarrow$
(i).
Here, $C - D^2\opL \affinec:=\{ z-D^2\opL \affinec \mid z \in C
 \}\subset \hilbertnewz$.
Suppose now that $\range (\opL \opM) =\hilbertnewz$.
Then, $D^2 \opL \opM$ is surjective, and it follows 
with Lemma \ref{lemma:implication_open_mapping_theorem}
that
(ii)
$\Rightarrow$
(iii)
$\Rightarrow$
(i)
$\Rightarrow$
$\exists\hat{x}\in\hilbertx$,
$D^2 \opL \opM\hat{x}\in \interior (C-D^2\opL \affinec)$
$\Rightarrow$
$\exists\hat{x}\in\hilbertx$,
$D^2 \opL (\opM\hat{x} + \affinec) \in \interior C$ 
$\Rightarrow$
(ii).
\end{proof}

Combining Proposition
\ref{proposition:positive_definite_case_necessary_condition}
and Lemma \ref{lemma:intKC_condition} gives the following corollary.

\begin{corollary}
\label{corollary:necessary_sufficient}
Let $\Psi :=|||\cdot|||$.
Assume that one of the following conditions are satisfied:
(i) $\affinec=0$,
(ii) $\range \opM = \hilbertnewz$, or
(iii) $\opA_2 \hat{x}=0$ for some $\hat{x}\in\hilbertx$.
Then, $F \in \Gamma_0(\hilbertx)$ if and only if
condition ($\spadesuit$) is satisfied. 
\end{corollary}
\begin{proof}
As $\opA_2 0 = \affinec$, (i) $\Rightarrow$ (iii).
Moreover, as $\opA_2 \hat{x}=0 \Leftrightarrow 
\opM \hat{x} = - \affinec \in\hilbertnewz$,
(ii) $\Rightarrow$ (iii).
Since $|||\cdot |||=\sigma_C$ for
$C:={\rm lev}_{\leq 1} |||\cdot |||_*$, it holds that
$||| 0  |||_*=0<1 \Leftrightarrow 0\in \interior C\neq \emptyset$.
Hence, (iii) of Corollary \ref{corollary:necessary_sufficient}
$\Rightarrow$
(iii) of Lemma \ref{lemma:intKC_condition}
$\Rightarrow \interior K_C\neq \emptyset$.
The assertion is thus verified by 
Proposition
\ref{proposition:positive_definite_case_necessary_condition}.
\end{proof}
Corollary \ref{corollary:necessary_sufficient} is useful
when $\Psi$ is a norm, because it gives simple ways
of seeing whether ($\spadesuit$) is necessary and sufficient.




\subsection{Proximal debiasing algorithms}\label{subsec:twin_prox}

We present iterative algorithms using the proximity operator
to compute the LiMES model
for the case of $\opD:=\gamma^{-1/2}I$ for simplicity,
which covers many applications including
the debiased sparse modeling
\eqref{eq:sparse_regression_mc_underdetermined},
SORR \eqref{eq:stable_regression},
ORR \eqref{eq:robust_estimation_mc_loss},
and robust classification \eqref{eq:mehinge_formulation} (see Section
\ref{subsec:robust_classification}).
(An extension to a general diagonal positive-definite operator $D$ is
straightforward.)
In this case,
\eqref{eq:fg_smooth_nonsmooth} reduces to
\begin{equation}
J_{\Omega_{\gamma}} =
 \underbrace{0.5\norm{\opA_1\cdot}^2 - \mu \hspace*{.2em} ^{\gamma}\Psi
\circ\opL\opA_2}_{{\rm smooth}} ~ + ~  \underbrace{\mu\Psi \circ \opA_2,}_{{\rm nonsmooth}}
\label{eq:fg_minimization_moreau}
\end{equation}
where
$\Omega_{\gamma}:=(\opA_1; \Psi_{\gamma^{-1/2}I}^{\opL}\circ \opA_2)$.
Here, 
the gradients of $0.5\norm{\opA_1\cdot}^2$
and $^{\gamma}\Psi \circ\opL\opA_2$ 
at $x\in\hilbertx$ are given, respectively, by
$\nabla (0.5\norm{\opA_1\cdot}^2)(x) = M_1^* \opA_1 x$
and (see Section \ref{subsec:notation})
\begin{align}
\hspace*{-2em} \nabla (^{\gamma}\Psi \!\circ\!\opL\opA_2)(x)
=&~ M_2^*\opL^*  \nabla\hspace*{.3em} ^{\gamma}\Psi
 (\opL\opA_2x)\nonumber\\
=&~ \gamma^{-1}M_2^*\opL^*  (I - \prox_{\gamma \Psi})
 (\opL\opA_2x).
\label{eq:moreau_composition_gradient}
\end{align}
Both gradient operators $\nabla (0.5\norm{\opA_1\cdot}^2)(x)$
and $\nabla (^{\gamma}\Psi \!\circ\!\opL\opA_2)$ are
Lipschitz continuous with constants
$\norm{M_1}^2$ and 
$\gamma^{-1}\norm{\opL}^2\norm{M_2}^2$, respectively.

\subsubsection{Proximal debiasing-gradient algorithm for typical type-S applications}
Let $\opA_2:=I$ which is used in typical type-S applications.
This allows to use an efficient algorithm requiring no auxiliary variable.
Specifically, under condition ($\spadesuit$),
\eqref{eq:fg_minimization_moreau} can be minimized
by the proximal gradient method:
\begin{align}
\hspace*{-1.7em} x_{k+1} \!:=&~ \prox_{\beta_k \mu\Psi} [x_k \!- \!
\beta_k (M_1^* \opA_1 x_k \nonumber\\
&\hspace*{1em}- \mu\gamma^{-1} \opL^*
(I - \prox_{\gamma\Psi})
 (\opL x))], ~ k\!\in\!\Natural,
\label{eq:pgm_moreau_reformulation}
\end{align}
where $\beta_k \in (0,2/(\norm{M_1}^2 + \mu \gamma^{-1} \norm{\opL}^2))$.
ISDA presented in Section \ref{subsubsec:isd_method} 
is reproduced by letting 
$\Psi:=\norm{\cdot}_1$ 
and $\opL:=P_{\mathcal{M}}$
in \eqref{eq:pgm_moreau_reformulation}, which makes
${\rm Prox}_{\delta \norm{\cdot}_1}={\rm soft}_\delta$
for any $\delta\in\real_{++}$.
The gradient term $\mu \opL^*\nabla \hspace*{.3em} ^{\gamma}\Psi(\opL x_k)$
actually plays the same role as the ``debiasing'' term of ISDA.
We therefore refer to the algorithm as
{\em the proximal debiasing-gradient algorithm}.

\subsubsection{Primal-dual debiasing algorithm for type-R applications}
\label{subsubsec:pdda_type_R}

Let $\tilde{\Psi}(z):= \mu \Psi(z + \affinec)$, $z\in\hilbertnewz$, so that 
$\tilde{\Psi}(\opM x)=\mu\Psi(\opA_2 x)$.
The problem in \eqref{eq:fg_minimization_moreau} can then be rewritten as 
\begin{equation}
\min_{x\in\hilbertx} ~0.5\norm{\opA_1 x}^2 - \mu \hspace*{.2em}^{\gamma}\Psi
(\opL\opA_2 x) ~ + ~ \tilde{\Psi}(\opM x).
\label{eq:fg_minimization_moreau_mod}
\end{equation}
By
$\prox_{\tilde{\Psi}/\sigma } (z) = - \affinec + 
 \prox_{\mu\Psi/\sigma}(z + \affinec)$
for $\sigma\in\real_{++}$,
 it follows that
\begin{align}
\hspace*{-.2em}\prox_{\sigma \tilde{\Psi}^*} (z) = &~ z - \sigma
 \prox_{\tilde{\Psi}/\sigma}(\sigma^{-1} z) \nonumber\\
= &~ z + \sigma \affinec -\sigma
 \prox_{\mu\Psi/\sigma}(\sigma^{-1} z + \affinec),
\end{align}
where the first equality is due to
the well-known identity
\cite[Theorem 14.3]{combettes}:
${\rm Prox}_{\gamma f} + \gamma{\rm Prox}_{f^*/\gamma}\circ
\gamma^{-1}I=I$ for any $f\in \Gamma_0(\hilbertnewz)$ and $\gamma\in\real_{++}$.
Problem \eqref{eq:fg_minimization_moreau_mod}
can be solved by
the existing operator splitting methods
such as the forward-backward-based primal-dual
method \cite{loris11,chen13,komodakis15}; see Algorithm \ref{alg:fbf}
below.\footnote{Due to the
presence of the Moreau envelope in the smooth part 
$0.5\norm{\opA_1\cdot}^2 -\hspace*{.2em} ^{\gamma}\Psi\circ \opA_2$, the popular ADMM and
Chambolle-Pock algorithms \cite{chambolle11} are not suitable to the present case,
because the former requires a minimizer of some function involving 
$0.5\norm{\opA_1\cdot}^2 - \hspace*{.2em}^{\gamma}\Psi\circ \opA_2$ (and thus requires an inner loop),
and the latter requires the proximity operator of 
$0.5\norm{\opA_1\cdot}^2 - \hspace*{.1em}^{\gamma}\Psi\circ \opA_2$ which cannot be written in a closed form in general.
Some other algorithms such as Condat's primal dual splitting method
\cite{condat13} may also be used.}
\begin{algorithm}[Primal-dual debiasing algorithm]
\label{alg:fbf}
~ \\
Set: $x_0\in\hilbertx$, $v_0\in\hilbertnewz$, $(\tau,\sigma)\in\real_{++}^2$,
$\beta_k\in \real_{++}$\\
For $k=0,1,2,\cdots$, do:\\
~~~$s_k = x_k - \tau [M_1^* \opA_1 x_k$\\
\hspace*{8em}$ -  \mu\gamma^{-1}\opM^*\opL^* (I - \prox_{\gamma\Psi})(\opL\opA_2  x_k)]$\\
~~~$u_k = s_k - \tau \opM^* v_k$\\
~~~$q_k = \prox_{\sigma \tilde{\Psi}^*}(v_k + \sigma \opM u_k)$\\
~~~$p_k = s_k - \tau \opM^* q_k$\\
~~~$(x_{k+1}, v_{k+1}) = (x_k,v_k) + \beta_k \left(
(p_k,q_k) - (x_k,v_k)
\right)$
\end{algorithm}
{\bf Convergence condition of Algorithm \ref{alg:fbf}:}
(i) $\tau\sigma \norm{\opM}^2\in(0,1)$ and $\tau\in(0,2/(\norm{M_1}^2 + \mu\gamma^{-1} \norm{\opL \opM}^2))$,
(ii) $(\beta_k)_{k\in\Natural}\subset(0,1]$ and $\inf_{k\in\Natural}
      \beta_k \in\real_{++}$,
(iii)
the function $J_{\Omega_{\gamma}}$  in
\eqref{eq:limes_model} has a minimizer, and 
(iv)
${\rm int}(\dom \tilde{\Psi})\cap \range \opM\neq \emptyset$.

\subsection{Stable outlier-robust regression as a special case of LiMES model}
\label{subsec:robust_stable_LiMES}

We consider a general situation when the augmented vector
$\xi_{\star}:=[x_{\star}^\top~\varepsilon_{\star}^\top]^\top\in\real^{n+m}$
obeys a zero-mean normal distribution with 
its (nonsingular) covariance matrix
$\Sigma_{\xi_{\star}}\in\real^{(n+m)\times (n+m)}$.
In this case, the standard statistical argument may suggest
the use of 
$0.5 \big\|\Sigma_\xi^{-1/2} \xi\big\|_2^2$,
where
$\xi:=[x^\top ~\varepsilon^\top]^\top\in\real^{n+m}$
and $\Sigma_\xi$ is an estimate of $\Sigma_{\xi_{\star}}$.
The estimate $\outputb - (Ax + \varepsilon)= \outputb -[\inputA~I_m]\xi$
of the sparse outlier is encouraged to be sparse by employing
$(\norm{\cdot}_1)_{\gamma^{-1/2}I}([\inputA~I_m]\xi - \outputb)$ 
as a fidelity function.
The above arguments amount to the following minimization problem:
\begin{equation}
 \min_{\xi\in\real^{n+m}} 
0.5\big\|
\underbrace{\Sigma_\xi^{-1/2} \xi}_{=:\opA_1 \xi}\big\|_2^2
+
~\mu
(\underbrace{\norm{\cdot}_1}_{=:\Psi})_{\gamma^{-1/2}I}(\underbrace{[\inputA~I_m]\xi
- \outputb}_{=:\opA_2 \xi}),
\label{eq:stable_regression_xi}
\end{equation}
which is a special case of the LiMES model 
with $\hilbertx:=\hilbertnewy:=\real^{n+m}$, $\hilbertnewz:=\real^{m}$,
$\Psi:=\norm{\cdot}_1$, $\opL:=I_m$, $\opD:=\gamma^{-1/2}I_m$,
and $\opA_2: \xi\mapsto [\inputA~I_m]\xi- \outputb$.
The formulation in \eqref{eq:stable_regression_xi} is
{\em a general form of SORR}.
Under the statistical assumption stated
in Section \ref{subsec:robust_regression},
it follows that
$\Sigma_{\xi_{\star}}=
{\rm diag}
( \sigma_{x_{\star}}^{2} I_n, \sigma_{\varepsilon_{\star}}^{2} I_m)$.
We therefore let
$\Sigma_\xi :=
{\rm diag}( \sigma_x^2 I_n , \sigma_{\varepsilon}^2 I_m)$,
with which \eqref{eq:stable_regression_xi} reduces to
\eqref{eq:stable_regression}.
Problem \eqref{eq:stable_regression_xi} can be solved 
by using Algorithm \ref{alg:fbf}
under the convexity condition in \eqref{eq:mu_condition_stable_regression}.
Note that, among the convergence conditions (i)--(iv) listed
below Algorithm \ref{alg:fbf},  only (i) and (ii) needs to be cared
in this specific case.
Indeed, conditions (iii) and (iv) are
satisfied automatically,
because \eqref{eq:stable_regression_xi} always has a solution
due to the coercivity of the objective function\footnote{A function $f\in\Gamma_0(\euclidspace)$ is {\em coercive}
if $f(x)\rightarrow +\infty$ as $\norm{x}\rightarrow +\infty$.}, and 
${\rm int}(\dom (\mu \norm{\cdot - y}_1)) \cap\range [A~I_m] = \real^m\neq \emptyset$.

 
\begin{table*}[t!]
\caption{LiMES Applications
($\mathcal{M}_1 := \range A^\top$ for sparse modeling,
$\mathcal{M}_1 := \range [I_n~I_n]^\top$ for SPCP, and
$M_{\rm RC}:=[y_1a_1\cdots y_m a_m]^\top$)}
\label{table:limes}

\centering
\begin{tabular}{|c|c|c|c|c|c|c|c|c|}
\hline
\!\!Application (type) \!\!\!\!& 
$\hilbertx$ & $\hilbertnewy$ &$\hilbertnewz$
& $\opA_1$ & $\Psi$ 
& $\opL$ & $D$ & $\opA_2$
\\ \hline
\!\!\!debiased\! sparse \!modeling (S)\!\!\!
& $\real^n$ & $\real^m$ & $\real^n$ 
& $A\cdot - y$ & $\norm{\cdot}_1$ 
& \!\!\!\!$P_{\mathcal{M}_1}$ \!\!\!\!\!\!
&  $\gamma^{-1/2}I_n$ & $I_n$  
\\ \hline
\!SORR (R)
&\!\! \!\hspace*{-.7em} $\real^{n+m}$\!\! &\!\!\!\hspace*{-.2em} $\real^{n+m}$\!\!\! & $\real^m$ 
& $\Sigma_{\xi}^{-1/2}$ & $\norm{\cdot}_1$ 
& $I_m$ &  $\gamma^{-1/2}I_m$ &\!\!\!\! $[A~I_m]\cdot - y$\!\!\!
\\ \hline
\!\!\!
\begin{tabular}{c}
\hspace*{-.5em}
\!\!\!\!SPCP (S)\!\!\!\! \\
\!\!\!\!\! \!\!\!\!
\hspace*{-1em}
\end{tabular}
\hspace*{-.7em}
\!\!\!&\hspace*{-.7em} $\real^{2n\times m}$\!\! &\!\!\hspace*{-.2em} $\real^{n\times m}$\!\! & \!\!$\real^{2n\times m}$\!\!
&\!\!\!\! $[I_n~I_n]\cdot - Y$\!\!\!
& $\begin{array}{c}
[L^\top~S^\top]^\top\mapsto \\
\!\!\mu_L\norm{L}_{\rm nuc}+\mu_S\norm{S}_1 \!\!\!\!\!\!\\
\end{array}$ \! \!
&\!\!\!\! $P_{\mathcal{M}_1}$ \!\!\!\!\!\!
&\!\!$
{\rm diag} \Big(\sqrt{\mu_L/\gamma} I_n, \sqrt{\mu_S/\gamma} I_n\Big)
 $\!\!
& $I_{2n}$
\\ \hline
robust classification (R)& 
$\real^{n}$ & $\real^{n}$ & $\real^m$ 
& $I_n$ & $\sigma_{[-1,0]^m}$ 
& $I_m$ &  $\gamma^{-1/2}I_m$ &\!\!\!\! $M_{\rm RC}\cdot - 1_m$\!\!\!\! 
\\ \hline
\end{tabular}

\end{table*}



\subsection{Stable principal component pursuit: A type-S application}
\label{subsec:spcp}

We consider the following model:
\begin{equation}
 Y =  L + S + W,
\label{eq:spcp_model}
\end{equation}
where $Y\in\real^{n\times m}$ is a noisy measurement of the superposition of the low-rank
matrix
$L\in\real^{n\times m}$ and the sparse matrix $S\in\real^{n\times m}$
with the additive white Gaussian noise
$W\in\real^{n\times m}$.
The problem of recovering $L$ and $S$ from the measurement $Y$ is called
{\it stable principal component pursuit (SPCP)} \cite{zhou10},
which can be formulated as follows:
\begin{equation}
  \min_{L,S\in\real^{n\times m}} 
0.5 \Big\|\underbrace{[I_n ~I_n]}_{=:M_1}
\left[\begin{array}{c}
L      \\
S
     \end{array}
\right]-\underbrace{Y}_{c_1}
\Big \|_{\rm F}^2
+ \Psi_\opD^{P_{\mathcal{M}_1}}
\left(
\left[\begin{array}{c}
L      \\
S
     \end{array}
\right]
\right).
\label{eq:spcp_formulation}
\end{equation}
Here,
$\norm{\cdot}_{\rm F}$ denotes the Frobenius norm,
$\opD  := 
{\rm diag}(\sqrt{\mu_L/\gamma} I_n, \sqrt{\mu_S/\gamma} I_n)
\in \! \real^{2n\times 2n},
~(\opL:=)$\\
$P_{\mathcal{M}_1}=
0.5[I_n~I_n]^\top[I_n~I_n]\in\real^{2n\times 2n}$
with $\mathcal{M}_1:=\range [I_n~I_n]^\top$,
and 
\begin{equation}
 \Psi: \real^{2n\times m} \rightarrow [0,+\infty):
\left[\begin{array}{c}
L      \\
S
     \end{array}
\right]
\mapsto
\mu_L\norm{L}_{\rm nuc} + \mu_S \norm{S}_1
\end{equation}
is a norm on $\real^{2n\times m}$
for any $\mu_L,\mu_S \in\real_{++}$
with $\norm{\cdot}_1$ 
summing up the absolute values of the entries.
It can be verified that
\begin{align}
\hspace*{-1.5em}\Psi_\opD^{P_{\mathcal{M}_1}}
\left(
\left[\begin{array}{c}
L      \\
S
     \end{array}
\right]
\right) = &~
\mu_L \left[\norm{L}_{\rm nuc} -
~^{\gamma}(\norm{\cdot}_{\rm nuc})\left(
\frac{L+S}{2}
\right)
\right]
  \nonumber \\
&\hspace*{-.1em} +
\mu_S \left[\norm{S}_1 -
~^{\gamma}(\norm{\cdot}_1)\left(
\frac{L+S}{2}
\right)
\right].
\label{eq:spcp_penalty}
\end{align}
The SPCP formulation given in \eqref{eq:spcp_formulation} is a special case of LiMES
for $\hilbertx:=\hilbertnewz:=\real^{2n\times m}$,
$\hilbertnewy:=\real^{n\times m}$, $\opA_1:=[I_n~I_n]\cdot - Y$,
and $\opA_2:=I_{2n}$ ($M_2:=I_{2n}$).
We emphasize here that
$(\opL:=)P_{\mathcal{M}_1}$ plays a key role for convexity preservation
as in Section
\ref{subsec:sparse_regression},
although the condition is given in terms of the
parameters contained in $\opD$ as shown in the following proposition.

\begin{proposition}[Convexity condition for \eqref{eq:spcp_formulation}]
\label{proposition:spcp_convexity}

Given $\mu_L,\mu_S,\gamma\in\real_{++}$, 
and $Y\in\real^{n\times m}$,
the smooth part
$0.5\norm{[I_n~I_n] \cdot - Y}_{\rm F}^2 -\hspace*{.1em} ^1(\Psi\circ \opD^{-1}) \circ \opD P_{\mathcal{M}_1}$ is convex
if and only if
 $\mu_L+\mu_S \leq 4\gamma$.
\end{proposition}
\begin{proof}
Since $\affinec:=0$ for SPCP, the smooth part of
\eqref{eq:spcp_formulation} is
convex if and only
if ($\spadesuit$) is satisfied by Corollary \ref{corollary:necessary_sufficient}.
It can be verified that
($\spadesuit$)~$\Leftrightarrow M_1^\top M_1 - \mu \opM^\top\opL^\top D^2 \opL \opM = 
[I_n\hspace*{.3em}I_n]^\top[I_n\hspace*{.3em}I_n] - P_{\mathcal{M}_1} \opD^2 P_{\mathcal{M}_1} =
\left(
1-\dfrac{\mu_L + \mu_S}{4\gamma}
\right)  
[I_n\hspace*{.3em}I_n]^\top[I_n\hspace*{.3em}I_n]\succeq O
\Leftrightarrow 4\gamma\geq \mu_L+\mu_S$.
\end{proof}

As the proximity operator of $\Psi$ can be computed directly by
those of the individual functions 
$\mu_L\norm{\cdot}_{\rm nuc}$ and $\mu_S\norm{\cdot}_1$,
the problem 
in \eqref{eq:spcp_formulation}
can be solved efficiently by the proximal gradient method
\eqref{eq:pgm_moreau_reformulation}.
We remark that the formulation in \eqref{eq:spcp_formulation}
for $\opL:=I$ has been studied in the framework of GMC
in \cite{yin19}, where the problem is solved
by a convex optimization algorithm involving dual variables.
In sharp contrast, no auxiliary variable is required
in our case, because $\opD$ is diagonal (cf.~Remark \ref{remark:separability_of_rmc}).
An $\ell_0$-based approach can also be found in the literature \cite{ulfarsson15}.


\subsection{Robust classification: A type-R application}
\label{subsec:robust_classification}


We consider a standard (supervised) classification task where
the pairs $(a_i,\outputb_i)\in\real^n\times \{+1,-1\}$,
$i\in\{1,2,\cdots,m\}$, of input vector and its label
are available.
We assume here that
the input vectors $a_i$ are normalized such that $\norm{a_i}_2=1$;
 it is implicitly assumed that $a_i\neq 0$.
We then consider the following problem formulation:
\begin{equation}
  \min_{x\in\real^{n}}~ 0.5\norm{x}_2^2 + \mu \sum_{i=1}^{m}
[\sigma_{[-1,0]}\circ (\outputb_i a_i^\top \cdot -
1)]_{\gamma^{-1/2}I}(x).
\label{eq:mehinge_formulation}
\end{equation}
Here, $\sigma_{[-1,0]}\circ (\outputb_i a_i^\top \cdot - 1)$ is the
popular hinge loss, and thus each summand is
the ME-hinge loss (see Example \ref{example:alime_loss}(b)).
To show that \eqref{eq:mehinge_formulation} is a special case of LiMES,
the following lemma will be used.

\begin{lemma}
\label{lemma:compositionA}
Let $\hilbertx$ and $\hilbertarbk$ be finite dimensional Hilbert spaces.
Let $\mathfrak{A}: \hilbertx\rightarrow \hilbertarbk: x \mapsto \genopL x + b$,
where $b\in \hilbertarbk$ and $\genopL:\hilbertx\rightarrow \hilbertarbk$
 is a bounded linear operator such that 
$\range \genopL = \hilbertarbk$ and
$\genopL^* \genopL = P_{\mathcal{V}}$
with $\mathcal{V}:=\range \genopL^* \subset\hilbertx$.
Then, for any $\psi\in\Gamma_0(\hilbertarbk)$ and $\gamma\in\real_{++}$,
it holds that
\begin{align}
 ^\gamma(\psi\circ\mathfrak{A}) =& ~ ^\gamma\psi\circ \mathfrak{A},\\
(\psi\circ \mathfrak{A})_{\gamma^{-1/2} I} =& ~ \psi_{\gamma^{-1/2} I}\circ \mathfrak{A}.
\label{eq:affine_composition_equivalence}
\end{align}
\end{lemma}

\begin{proof}
See Appendix \ref{subsec:proof_lemma_compositionA}.
\end{proof}

\begin{proposition}[\eqref{eq:mehinge_formulation}
as a special case of LiMES model]
\label{proposition:me_hinge_equivalence}

Let
$\Psi:\real^m\rightarrow \real: z:=[z_1,z_2,\cdots,z_m]^\top \mapsto 
\sigma_{[-1,0]^m}(z)
= \sum_{i=1}^{m} \sigma_{[-1,0]}(z_i)$
and 
$\opA_2: \real^n\rightarrow \real^m:
x \mapsto \opM x - 1_m$ with
$\opM:=[\outputb_1 a_1 ~\outputb_2 a_2 ~\cdots~\outputb_m a_m]^\top
\in\real^{m\times n}$
and  $1_m:=[1,1,\cdots,1]^\top\in\real^m$.
Then, the second term in \eqref{eq:mehinge_formulation} can be expressed as
\begin{equation}
 \Psi_{\gamma^{-1/2}I}\circ \opA_2
= \sum_{i=1}^{m}
[\sigma_{[-1,0]}\circ (\outputb_i a_i^\top \cdot - 1) ]_{\gamma^{-1/2}I}.
\end{equation}

\end{proposition}

\begin{proof}
Let $(O\neq)M_{2,i} : \real^n\rightarrow \real:x\mapsto \outputb_i 
a_i^\top x$, $i=1,2,\cdots,m$.
It then holds that $M_{2,i}^* M_{2,i} = P_{\range M_{2,i}^*}$ as
$\norm{M_{2,i}}=1$,
and $\range M_{2,i} =\real$ as $M_{2,i} \neq O$.
For each $i\in\{1,2,\cdots,m\}$,
letting $\hilbertarbk:=\real$,
$\psi:=\sigma_{[-1,0]}$, $\genopL:=M_{2,i}$, and $b:=-1$
in  Lemma \ref{lemma:compositionA} yields
\begin{equation*}
\label{eq:equivalence_proposition7}
 (\sigma_{[-1,0]})_{\gamma^{-1/2}I} (\outputb_i a_i^\top x - 1)=
[\sigma_{[-1,0]}\circ (\outputb_i a_i^\top \cdot - 1) ]_{\gamma^{-1/2}I} (x),
\end{equation*}
from which together with the separability of $\Psi$
it follows that
$\Psi_{\gamma^{-1/2}I}\circ \opA_2 (x)
= \sum_{i=1}^{m} 
(\sigma_{[-1,0]})_{\gamma^{-1/2}I} (\outputb_i a_i^\top x - 1)=
\sum_{i=1}^{m}
[\sigma_{[-1,0]}\circ (\outputb_i a_i^\top \cdot - 1) ]_{\gamma^{-1/2}I}(x)$.
\end{proof}

In light of Proposition \ref{proposition:me_hinge_equivalence},
the formulation in \eqref{eq:mehinge_formulation}
is a special case of LiMES
for $\hilbertx:=\hilbertnewy:=\real^n$,
$\hilbertnewz:=\real^m$, 
$\opA_1:=I_n$ ($M_1:=I_n$),
$\opL:=I_m$,
and $D:=\gamma^{-1/2}I_m$.
Table \ref{table:limes} summarizes the applications of LiMES.
The convexity condition is given as below.

\begin{proposition}[Convexity condition for \eqref{eq:mehinge_formulation}]
\label{proposition:me_hinge_convexity}
 The smooth part of \eqref{eq:mehinge_formulation}
is convex if $ \mu \lambda_{\max}(\opM^\top \opM) \leq \gamma$.
Suppose, in particular, that 
(i) $\range \opM = \real^m$, or 
(ii) $\gamma\in (1, +\infty)$.
Then, the smooth part of \eqref{eq:mehinge_formulation}
is convex if and only if 
$ \mu \lambda_{\max}(\opM^\top \opM) \leq \gamma$.

\end{proposition}

\begin{proof}
Assume that $\range \opM = \hilbertnewz (=\real^m)$.
In this case, 
since $D$ is a positive definite operator and $\opL=I_m$,
we have $\range (D^2 \opL \opM) = \hilbertnewz$, and hence
(iii) of Lemma \ref{lemma:intKC_condition} is clearly satisfied.
Assume on the other hand that $\gamma\in(1,+\infty)$.
It then holds that $D^2 \opL\opA_2 0_n=\gamma^{-1}(\opM 0_n -1_m) 
= - \gamma^{-1} 1_m\in(-1,0)^m
=\interior C$, 
and thus
(iii) of Lemma \ref{lemma:intKC_condition} is satisfied again.
Thus, it follows that $\interior K_C\neq \emptyset$
under any of conditions (i) and (ii) of the proposition.
Hence, in light of Propositions \ref{proposition:positive_definite_case_necessary_condition}
and \ref{proposition:me_hinge_equivalence},
 the smooth part of \eqref{eq:mehinge_formulation} is convex 
if and only if
($\spadesuit$) is satisfied.
Finally,
($\spadesuit$)
$\Leftrightarrow M_1^\top M_1 - \mu \opM^\top\opL^\top D^2 \opL \opM = I_n -  \mu \gamma^{-1} \opM^\top \opM
 \succeq O \Leftrightarrow 1-\mu\gamma^{-1} \lambda_{\max}(\opM^\top
 \opM)\geq 0$.
This verifies the assertion.
\end{proof}

\section{Numerical Examples}\label{sec:numerical}

We show the efficacy of the LiMES model in two applications:
sparse modeling in the underdetermined case and
robust regression.

\subsection{Experiment A: Sparse modeling in underdetermined case}
\label{subsec:exp_sparse_regression}
We compare the performance of the PMC penalty (see Section
\ref{subsec:sparse_regression}) for sparse modeling with those of the
following penalties: 
$\ell_{1}$ (lasso) implemented by
the iterative shrinkage-thresholding algorithm (ISTA) \cite{beck09}, and GMC with 
the linear operator $B := (\alpha_{\textrm{GMC}}/\mu)^{1/2}A$ for $\alpha_{\textrm{GMC}}\in[0,1]$.
The standard linear model 
$y = Ax_{\diamond} + \varepsilon_{\star}$
is considered
with the i.i.d.~standard Gaussian input matrix 
$A \in \mathbb{R}^{m \times n}$ for $m := 64$ and $n := 128$.
Here, 
$x_{\diamond} \in \mathbb{R}^{n}$
is the sparse unknown vector with $s$ nonzero components,
and $\varepsilon_{\star} \in \mathbb{R}^{m}$ is 
the i.i.d.~zero-mean Gaussian noise vector with signal-to-noise
ratio (SNR) 20 dB and 30 dB, where $\mbox{SNR} :=
\|Ax_{\diamond}\|^{2}_{2}/\|\varepsilon_{\star}\|^{2}_{2}$.
The regularization parameter is tuned
so that all the methods share the same sparseness
as the true $x_\star$ 
with respect to the sparseness measure \cite{hoyer04}
$[n/(n-\sqrt{n})]
\left[1-\left\|x\right\|_{1}/(\sqrt{n}\left\|x\right\|_{2})\right]
\in[0,1]$.
For PMC, $\gamma := \mu/[\alpha_{\textrm{PMC}}
\lambda_{\min}^{++}(A^\top A)]$ for $\alpha_{\textrm{PMC}} \in (0,1]$
is used (see Section \ref{subsec:sparse_regression}).
The parameters $\alpha_{\textrm{GMC}}$ and $\alpha_{\textrm{PMC}}$
are tuned manually to attain the lowest system mismatch for each method.
The results are averaged over 300 trials.


Figures \ref{fig:expA_sparse_regression}(a) 
and \ref{fig:expA_sparse_regression}(b) 
show the system  mismatch $\|x_{\diamond} -
x\|_{2}^{2}/\|x_{\diamond}\|_{2}^{2}$
for different sparsity levels.
It can be seen that PMC outperforms the other methods
particularly when the sparsity level is middle, 
$s\in[18,24]$ more specifically.
Note here that the proposed approach requires no auxiliary vector
unlike GMC (see Remark \ref{remark:separability_of_rmc}).
Figure \ref{fig:expA_sparse_regression}(c) plots the average estimate of
each method over the 300 trials for SNR 20 dB with sparsity level
$s:=21$.
It can be seen that PMC estimates $x_{\diamond}$ with high accuracy,
indicating that the estimation bias is reduced successfully.

Finally, Fig.~\ref{fig:expA_instability_instance} shows
a particular instance (SNR 20 dB, $s:=21$)
to show that a direct application of 
the original MC penalty to an underdetermined system may fail.
The MC penalty is 
implemented by
ISDA in \eqref{eq:isda}.
For reference, the performances of the ordinary least square (OLS) estimate
$A^{\dagger}y\in\argmin_{x\in\real^n}\norm{Ax-y}_2^2$ 
and the ridge regression are plotted.
Due to the nonconvexity of the objective function 
involving the original MC penalty in the present underdetermined case,
the system mismatch of MC could be unacceptably large sometimes,
although it may perform better than PMC on average.
This clearly suggests the efficacy of the PMC penalty.

\begin{figure}[t!]
\centering
\begin{tabular}{cc}
\begin{minipage}{4cm}
 \subfigure[SNR 20 dB]{
\hspace*{-1.3em}  \includegraphics[height=4cm]{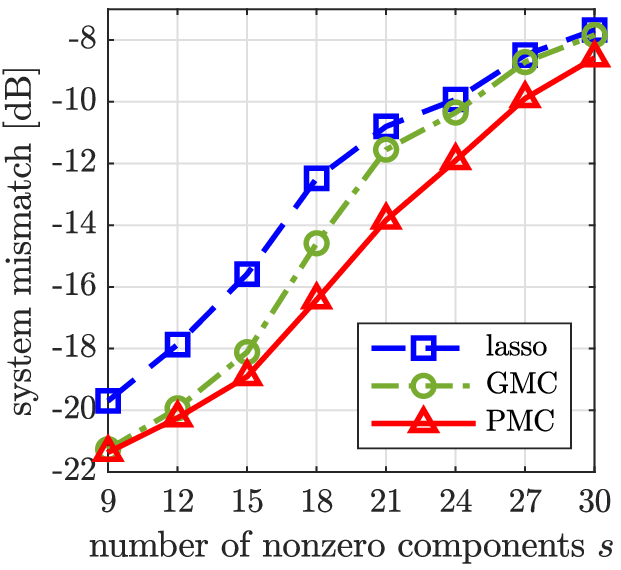}
 }\vspace*{-1em} 
\end{minipage}
 &
\begin{minipage}{4cm} 
\subfigure[SNR 30 dB]{
\hspace*{-1.3em} \includegraphics[height=4cm]{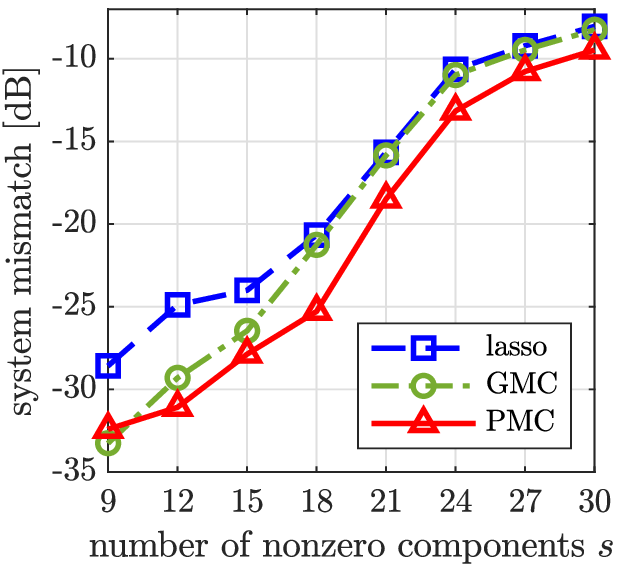}
}\vspace*{-1em}
\end{minipage}
 \\
\end{tabular}
\subfigure[average estimates for SNR $20$ dB ($s:=21$)]{
\includegraphics[height=3.5cm]{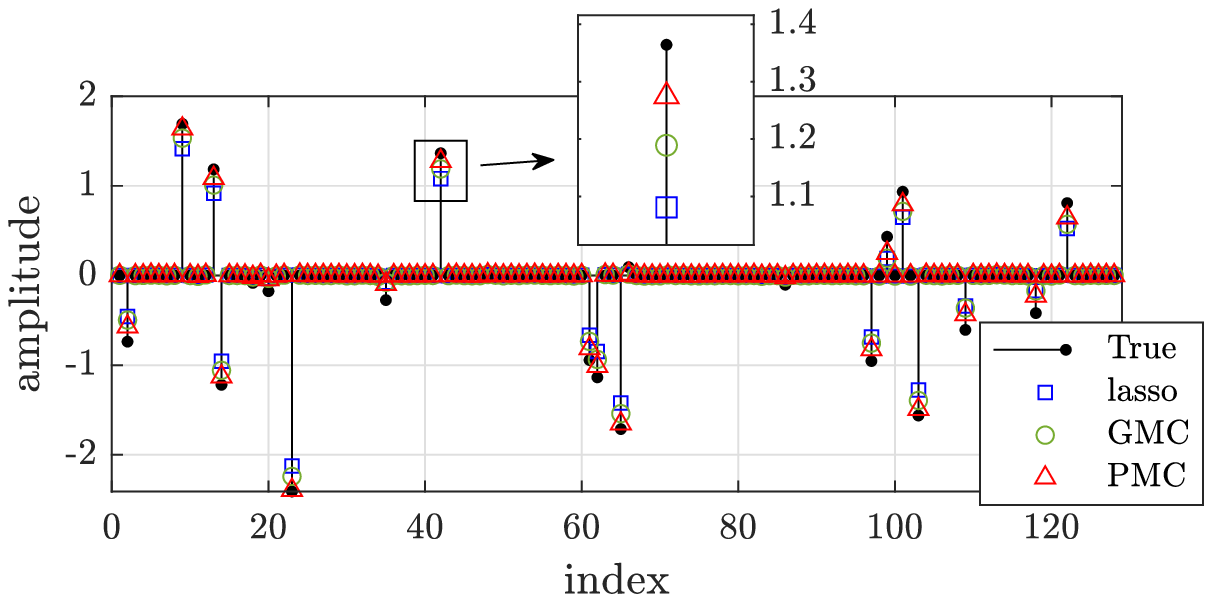}
}\vspace*{-1em}
\caption{Experiment A: Learning curves and the average estimates.}
\label{fig:expA_sparse_regression}
\end{figure}

\begin{figure}[t!]
\centering
\includegraphics[height=3.2cm]{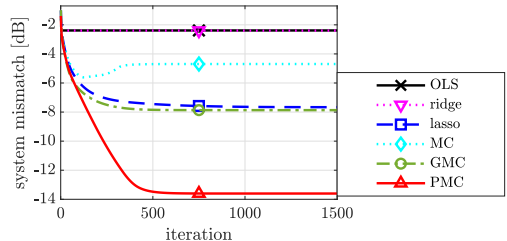}
\caption{Experiment A: A particular instance in which
a direct application of the MC penalty to an underdetermined system fails.}
\label{fig:expA_instability_instance}
\end{figure}


\subsection{Experiment B: Robust regression in the presence of outlier}

\begin{figure}[t!]
  \centering

\begin{tabular}{cc}
\begin{minipage}{4cm}
 \subfigure[SNR 10 dB,  SOR $-30$ dB]{
\hspace*{-1.3em}  \includegraphics[height=4cm]{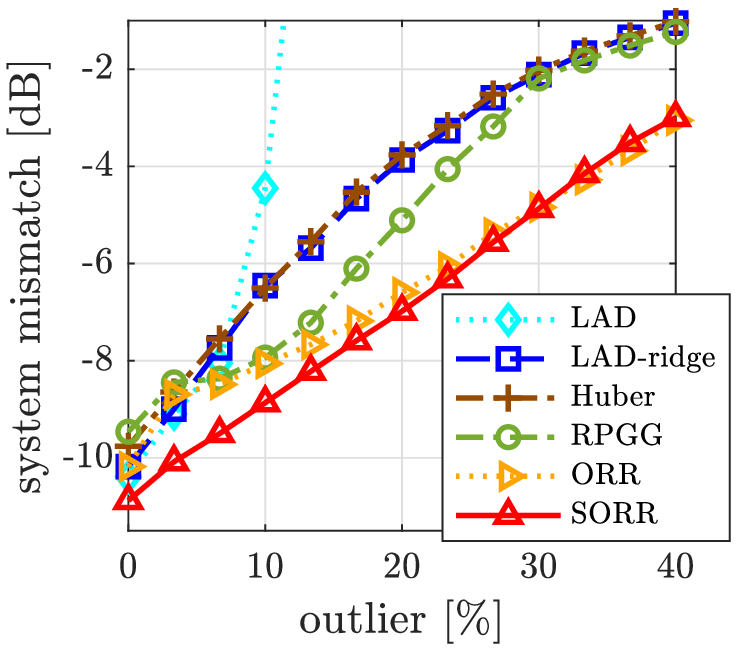} 
 }\vspace*{-1em} 
\end{minipage}
 &
\begin{minipage}{4cm} 
\subfigure[SNR 20 dB,  SOR $-40$ dB]{
\hspace*{-1.3em}  \includegraphics[height=4cm]{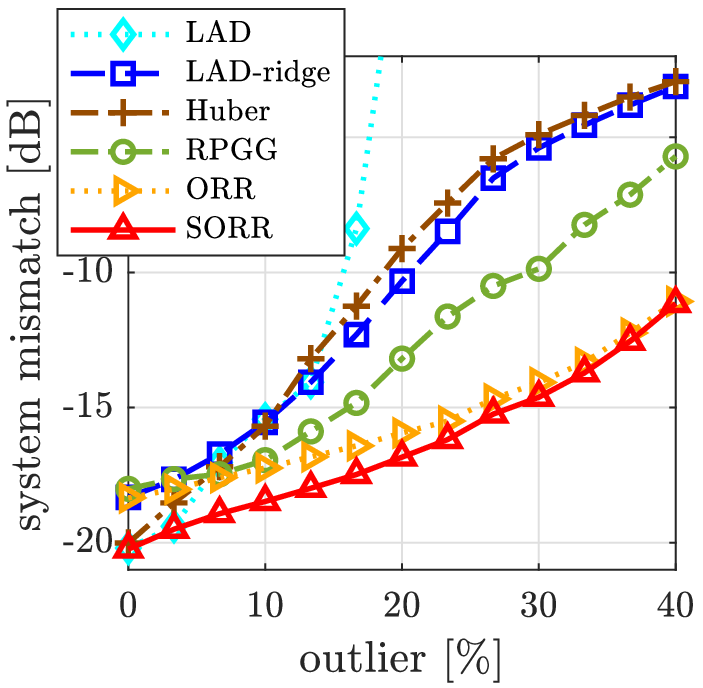} 
}\vspace*{-1em}
\end{minipage}
 \\
\end{tabular}

  \caption{Experiment B: System mismatch across outlier density.}
  \label{fig:expB_diff_outlier}
\end{figure}

\begin{figure}[t!]

\begin{tabular}{cc}
\begin{minipage}{4cm}
  \centering
 \subfigure[SNR 10 dB, outlier 15 \%]{
\hspace*{-1.5em}   \includegraphics[height=3.8cm]{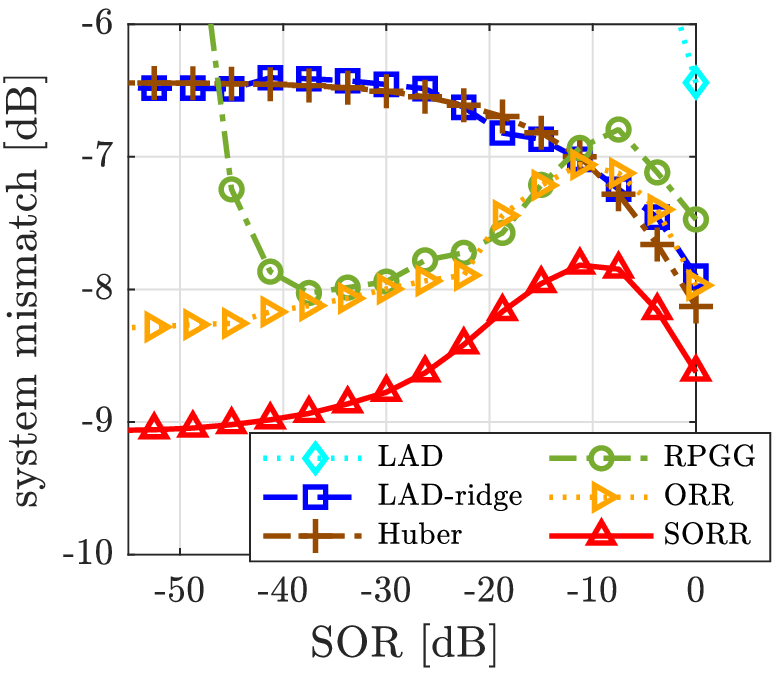} 
 }\vspace*{-1em} 
\end{minipage}
 &
\begin{minipage}{4cm} 
  \centering
\subfigure[SNR 20 dB, outlier 10 \%]{
\hspace*{-1.8em}   \includegraphics[height=3.8cm]{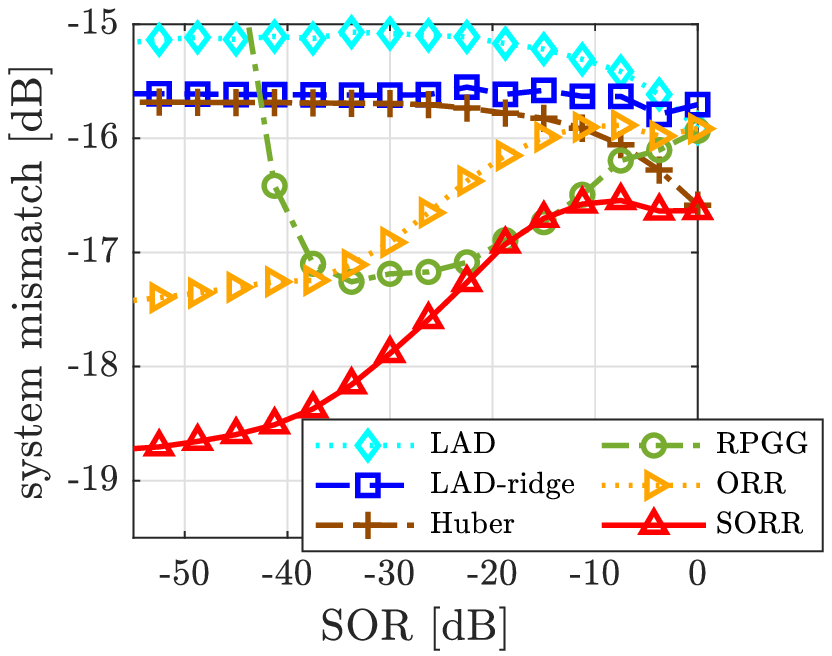} 
}\vspace*{-1em}
\end{minipage}
 \\
\end{tabular}
  \caption{Experiment B: System mismatch across SOR.}
  \label{fig:expB_diff_sor}
\end{figure}

We compare the performances of 
SORR and ORR 
(see Section \ref{subsec:robust_regression})
for robust regression with those of
LAD \cite{huber_book},
LAD-ridge ($\ell_1$-loss $+$ Tikhonov regularization),
Huber's loss 
\hspace*{-.2em}$^{\gamma}\norm{\cdot}_1$ \cite{zoubir_book18,huber_book},
and the state-of-the-art method called
the robust projected generalized gradient (RPGG) algorithm
\cite{yang19} which is based on the following formulation\footnote{
Although RPGG is a method for robust sparse recovery, 
it could be used in the present nonsparse case by letting $\mu:=0$.
We instead tune the $\mu$ to seek for its potentially better
performances.
The MC function is employed in our simulations for both data fidelity and penalty,
as in the simulations of \cite{yang19}.
}:
$\min_{x \in \mathbb{R}^n, e\in\real^m } \mu (\|\cdot \|_1)_{\gamma_1^{-1}I}
	 (x) +  (\|\cdot \|_1)_{\gamma_2^{-1}I} (e)$ 
subject to $ y = A x+e$
for $\gamma_1,\gamma_2\in(0,+\infty]$.
The sparse outlier model
$\outputb:= A x_{\star} + \varepsilon_{\star} + o_{\diamond}$
is used, where
the input matrix $A\in\real^{m\times n}$ and noise
$\varepsilon_{\star}\in\real^m$ are generated randomly with 
$m:=128$ and $n:=64$ in the same way as in Experiment A.
To show that SORR is stable under large Gaussian noise,
we consider the cases of SNR 10 dB and 20 dB.
The nonsparse vector $x_{\star}\in\real^n$ is generated
randomly from the i.i.d.~standard Gaussian distribution (i.e.,
$\sigma_{x_{\star}}^2:=1$).
The outlier vector $o_{\diamond}$ is sparse
with nonzero positions chosen randomly and with nonzero components
generated from an i.i.d.~zero-mean Gaussian distribution
with variance determined by 
the signal-to-outlier ratio (SOR)
$(\|Ax_{\star}\|_2^2/m) [\|o_{\diamond}\|_2^2/ {\rm
supp}(o_{\diamond})]^{-1}$.
Here, ${\rm supp}(x):=\abs{\{i\in\{1,2,\cdots,m\}\mid x_i\neq 0\}}$
is the support of a vector $x\in\real^m$.
For SORR,
$\sigma_x^2:=\sigma_{x_{\star}}^2$ and $\sigma_{\varepsilon}^2:=\sigma_{\varepsilon_{\star}}^2$ 
are used to show the potential performance.
For the primal-dual debiasing algorithm, the parameters are chosen as follows.
The parameters $\tau$ and $\sigma$ are set to
slightly smaller values than the upper bounds, respectively,
shown under Algorithm \ref{alg:fbf}.
We simply let $\beta_k:=1$ for all $k\in\Natural$, and
tune $\gamma$ and $\mu$ based on Proposition
\ref{proposition:orsr_convexity} by grid search
to attain the best performance.
For RPGG,
we let $\gamma_1:=+\infty$ (i.e.,
$(\|\cdot\|_1)_{\gamma_1^{-1}I}=\norm{\cdot}_1$)
as $x_{\star}$ is nonsparse,
and tune $\mu$ and $\gamma_2$ as well as the step size
by grid search.
For the other methods involving regularizers, the regularization parameters are
tuned by grid search  to attain the best performance.
For Huber's loss, $\gamma$ is chosen
to attain the best performance.
The results are averaged over 300 trials.

Figure \ref{fig:expB_diff_outlier} plots the results across
outlier density ${\rm supp}(o_{\diamond})/m$.
The proposed SORR method exhibits highly accurate and stable
performances, and it outperforms all the other methods significantly.
To be specific, the difference from ORR is notable
when the outlier density is low to middle.
It should be mentioned that LAD performed poorly due to the presence of
heavy noise as well as strong outliers.
Figure \ref{fig:expB_diff_sor} plots the results across
SOR to show the impacts of the change of the outlier power
on the performance.
Remarkably, the performances of SORR and ORR even improve
as SOR decreases below $-12$ dB.
This is because the influence of huge outliers on the MC loss 
vanishes above a certain range due to the same reason as for Tukey's
loss \cite{huber_book} and because such huge outliers will be easier to
detect at the same time.
The results clearly indicate the remarkable robustness of SORR (and ORR)
against huge outliers.
We mention that RPGG also exhibits a similar tendency over a reasonable range,
although its performance degrades for SOR below $-40$ dB.\footnote{
When SOR is small, the initial error of the outlier vector is large, and
this increases the number of iterations for the RPGG algorithm
to reach a sufficiently small
error. The step size is therefore chosen to be large so that
the algorithm converges in a comparable number of iterations to the
other methods, and this is the reason for the sharp rise of the errors
observed in Fig.~\ref{fig:expB_diff_sor}.
One may suppress it by decreasing the step size, but this then
results in slow convergence, causing an undesirable increase of complexity.
}



\section{Concluding Remarks}\label{sec:conclusion}

We presented the efficient framework based on the LiMES model.
The PMC penalty composes the Moreau envelope contained in the standard MC penalty
with the projection operator onto the input subspace,
thereby restricting the Moreau-enhancement effect
to the subspace for preserving the overall convexity
even in the underdetermined case.
SORR distinguishes Gaussian noise and sparse outlier explicitly
to attain stable performances in highly noisy situations.
The convexity conditions for those specific instances
were discussed in a unified fashion with the LiMES model.
While the LiMES function is ``nonseparable'',
the objective function involved in  the Moreau envelope is
``separable''.
This {\em mixed nature of separability and nonseparability}
allows an application of the LiMES model to the case
when the fidelity term is not strongly convex
(as in the underdetermined case of linear regression)
with an efficient implementation using the proximal gradient method.
The operators $\opL$ and $\opA_2$ play key roles in the model:
$\opL$ corresponds to the projection mentioned above and
$\opA_2$ takes care of robust regression.
The proximal debiasing algorithms to compute the LiMES model
require convexity of the smooth part of the objective function,
for which a sufficient condition was presented.
The condition was shown to be a necessary condition as well
under the nonempty-interior assumption
when the seed function is a support function.
This is the case for instance when
the seed function is a norm and
the range of $\opA_2$ contains the zero vector.
Applications of the LiMES model
to SPCP and robust classification
were also presented.
The hinge loss function widely used for robust classification 
was shown to be expressed as a composition of the support function
of a closed interval $[-1,0]$ and an affine operator.
Numerical examples showed that (i) the PMC penalty achieved debiased
sparse modeling for underdetermined systems
as well as outperforming GMC, and 
that (ii) SORR achieved stable and remarkably robust performances 
in the presence of both heavy Gaussian noise and 
sparse outlier as well as outperforming the existing robust methods
including LAD, Huber's loss, and RPGG.

The LiMES model will serve as a powerful tool to
enhance performances with respect to a variety of penalty/loss functions 
based on the solid foundation of convex analysis,
and there are plenty of opportunities to explore its further applications.
In particular, it is our future works to investigate the efficacy of 
the LiMES model in SPCP and robust classification.





\bibliographystyle{IEEEtran}
\bibliography{weaklyconvex,nmf}

\appendices


\newcounter{appnum}
\setcounter{appnum}{1}

\setcounter{theorem}{0}
\renewcommand{\thetheorem}{\Alph{appnum}.\arabic{theorem}}

\setcounter{lemma}{0}
\renewcommand{\thelemma}{\Alph{appnum}.\arabic{lemma}}
\setcounter{example}{0}
\renewcommand{\theexample}{\Alph{appnum}.\arabic{example}}
\setcounter{equation}{0}
\renewcommand{\theequation}{\Alph{appnum}.\arabic{equation}}
\setcounter{claim}{0}
\renewcommand{\theclaim}{\Alph{appnum}.\arabic{claim}}
\setcounter{remark}{0}
\renewcommand{\theremark}{\Alph{appnum}.\arabic{remark}}


\section{Proof of Proposition  \ref{proposition:rmc_convexity}}
\label{subsec:proof_pmc}

Since $c:=0$ for the debiased sparse modeling, the smooth part of
\eqref{eq:sparse_regression_mc_underdetermined} is
convex if and only
if ($\spadesuit$) is satisfied
by Corollary \ref{corollary:necessary_sufficient}.
By definition of
$\mathcal{M}:=\range \inputA^\top$, moreover,
it holds that
$P_{\mathcal{M}}A^\top =A^\top$,
from which together with 
$P_{\mathcal{M}}=P_{\mathcal{M}}\circ P_{\mathcal{M}}$
it follows that
($\spadesuit$) 
$\Leftrightarrow M_1^\top M_1 - \mu \opM^\top\opL^\top D^2 \opL \opM =
\inputA^\top \inputA-\mu\gamma^{-1}P_{\mathcal{M}}
= P_{\mathcal{M}} (\inputA^\top \inputA -\mu\gamma^{-1} I) P_{\mathcal{M}}
\succeq O\Leftrightarrow \lambda_{\min}^{++}(\inputA^\top \inputA) \geq
\mu\gamma^{-1}$.
\migip

\setcounter{appnum}{2}
\setcounter{equation}{0}

\section{Proof of Proposition  \ref{proposition:orsr_convexity}}
\label{subsec:proof_orsr}

According to the discussions in Section
\ref{subsec:robust_stable_LiMES},
\eqref{eq:stable_regression} is equivalent to \eqref{eq:stable_regression_xi}.
Since $\range \opM =\range [\inputA~I_m]=\hilbertnewz$ for SORR, 
the smooth part of \eqref{eq:stable_regression_xi} is
convex if and only
if ($\spadesuit$) is satisfied
by Corollary \ref{corollary:necessary_sufficient}.
We prove the equivalence
($\spadesuit$) $\Leftrightarrow$ \eqref{eq:mu_condition_stable_regression}
below.

For $\Sigma_\xi^{-1}=
{\rm diag} (\sigma_x^{-2} I_n,\sigma_{\varepsilon}^{-2} I_m)
$, 
it holds that
($\spadesuit$) $\Leftrightarrow M_1^\top M_1 - \mu \opM^\top\opL^\top D^2 \opL \opM =\Sigma_\xi^{-1} -
\mu\gamma^{-1}[\inputA~I_m]^\top[\inputA~I_m]\succeq O$
which can be expressed equivalently as follows:
\begin{equation}
 \left[
\begin{array}{cc}
\mu^{-1}\gamma\sigma_x^{-2}I_n - \inputA^\top\inputA &- \inputA^\top \\
 - \inputA &(\mu^{-1} \gamma\sigma_{\varepsilon}^{-2} -1)I_m \\
\end{array}
\right]\succeq O.
\label{eq:convexitycondition4stableregression}
\end{equation}
By \cite[Theorem 7.7.9]{horn_johnson13},
\eqref{eq:convexitycondition4stableregression} holds if and only if
all of the following conditions are satisfied:
\begin{enumerate}
 \item[(i)]
 $\mu^{-1}\gamma\sigma_x^{-2}I_n - \inputA^\top\inputA\succeq
       O~(\Leftrightarrow \mu \lambda_{\max}(\inputA^\top\inputA)\leq \gamma\sigma_x^{-2})$;
 \item[(ii)] 
 $(\mu^{-1} \gamma\sigma_{\varepsilon}^{-2} -1)I_m \succeq O ~(\Leftrightarrow
       \mu\leq \gamma\sigma_{\varepsilon}^{-2})$;
 \item[(iii)]  $-\inputA^\top=
(\mu^{-1}\gamma\sigma_x^{-2}I_n - \inputA^\top\inputA)^{1/2}
\Upsilon
((\mu^{-1} \gamma\sigma_{\varepsilon}^{-2} -1)I_m)^{1/2}$
for some $\Upsilon\in\real^{n\times m}$ with its largest singular value at most one.
\end{enumerate}
If $\inputA=O$, then conditions (i) and (iii) hold trivially, and condition (ii)
 coincides with \eqref{eq:mu_condition_stable_regression}.
Assume that $\inputA\neq O$ in the following.
We shall show below that
(i)--(iii) $\Leftrightarrow$ \eqref{eq:mu_condition_stable_regression}.
Suppose that conditions (i)--(iii) are satisfied.
Condition (iii) under $\inputA\neq O$ implies that $\mu^{-1} \gamma\sigma_{\varepsilon}^{-2}
 -1\neq 0$, and hence $\mu^{-1} \gamma\sigma_{\varepsilon}^{-2} -1> 0$ by condition (ii).
The equality in condition (iii) above can be rewritten as
\begin{equation}
\nu_{\varepsilon}
\inputA^\top=
(\nu_xI_n - \inputA^\top\inputA)^{1/2}
\tilde{\Upsilon},
\label{eq:condition3a}
\end{equation}
where $\nu_{\varepsilon}:=(\mu^{-1} \gamma\sigma_{\varepsilon}^{-2} -1)^{-1/2}> 0$,
$\nu_x:=\mu^{-1}\gamma\sigma_x^{-2}> 0$, and $\tilde{\Upsilon}:=-\Upsilon$.
Let $\inputA=V\Sigma U^\top$ be a singular value decomposition of $\inputA$,
where $U\in\real^{n\times n}$ and $V\in\real^{m\times m}$ are 
orthogonal  matrices, and
$\Sigma=
{\rm diag}(\varsigma_{1},\varsigma_{2},\cdots,\varsigma_{\min\{n,m\}})
 \in\real^{m\times n}$
having $ \varsigma_{1}\geq\varsigma_{2} \geq\cdots \geq
\varsigma_{\min\{n,m\}} \geq 0$ for the diagonal entries
and zeros for the off-diagonal entries.
Then, \eqref{eq:condition3a} can be rewritten as
\begin{align}
U (\nu_{\varepsilon} \Sigma^\top) V^\top =&~
U ( \nu_x I_n - \Sigma^\top\Sigma)^{1/2} U^\top 
\tilde{\Upsilon}\nonumber\\
\Leftrightarrow 
\nu_{\varepsilon} \Sigma^\top =&~
 ( \nu_x I_n - \Sigma^\top\Sigma)^{1/2} U^\top 
\tilde{\Upsilon}V.
\label{eq:condition3b}
\end{align}
Let $\tilde{\Upsilon}= -\Upsilon= U \Xi V^\top$
for some matrix $\Xi\in\real^{n\times m}$.
Then, \eqref{eq:condition3b} reads
\begin{align}
\nu_{\varepsilon} \Sigma^\top =
 ( \nu_x I_n - \Sigma^\top\Sigma)^{1/2}
\Xi.
\label{eq:condition3c}
\end{align}
Noting that $\varsigma_1>0$ due to the assumption $\inputA\neq O$,
one can verify from \eqref{eq:condition3c}
that $\Xi$ must be written in the following form:
\begin{equation}
 \Xi=
{\rm diag}(\varsigma_{1,\Upsilon}, \Xi_{2,2})
\in\real^{n\times m},
\label{eq:condition3d}
\end{equation}
of which the $(1,1)$ entry is $\varsigma_{1,\Upsilon}>0$,
the lower-right submatrix is $\Xi_{2,2}\in\real^{(n-1)\times (m-1)}$,
and the entries of the off-diagonal blocks are zeros.
By \eqref{eq:condition3c} and \eqref{eq:condition3d}, we obtain
\begin{equation}
 \nu_{\varepsilon} \varsigma_1 = (\nu_x- \varsigma_1^2)^{1/2} \varsigma_{1,\Upsilon},
\label{eq:condition3e}
\end{equation}
where $\nu_x - \varsigma_1^2>0$ as $\nu_{\varepsilon}\varsigma_1>0$.
To see that $\varsigma_{1,\Upsilon}$ is a singular value of 
$\Upsilon$ (or that of $\tilde{\Upsilon}$ equivalently),
let $\Xi_{2,2}:=V_{\Xi_{2,2}} \Sigma_{\Xi_{2,2}} U_{\Xi_{2,2}}^\top$
be a singular value decomposition of 
$\Xi_{2,2}$, where 
$V_{\Xi_{2,2}}\in\real^{(n-1)\times (n-1)}$
and 
$U_{\Xi_{2,2}}^\top\in\real^{(m-1)\times(m-1)}$ are orthogonal
 matrices,
and 
$\Sigma_{\Xi_{2,2}}
:=
{\rm diag}(\varsigma_{2,\Upsilon}, \varsigma_{3,\Upsilon}, \cdots,\varsigma_{\min\{n,m\},\Upsilon})
 \in\real^{(n-1)\times (m-1)}$
for singular values $\varsigma_{i,\Upsilon}\geq 0$
for $i\in \{2,3,\cdots,\min\{n,m\}\}$.
It then follows that
$\Xi = V_{\Xi}\Sigma_{\Xi} U_{\Xi}^\top$,
where 
$V_{\Xi}:=
{\rm diag}(1,V_{\Xi_{2,2}})
$,
$U_{\Xi}:=
{\rm diag}(1,U_{\Xi_{2,2}})
$, and
$\Sigma_{\Xi}:=
{\rm diag}(\varsigma_{1,\Upsilon},\Sigma_{\Xi_{2,2}})
$.
Thus, $\Upsilon= U_{\Upsilon} \Sigma_{\Upsilon} V_{\Upsilon}^\top$
gives a singular value decomposition of $\Upsilon$
with $U_{\Upsilon} := -U U_{\Xi}$,
$\Sigma_{\Upsilon}:= \Sigma_{\Xi} $, and
$V_{\Upsilon}:=  V V_{\Xi}$,
where $U_{\Upsilon}$ and $V_{\Upsilon}$ are
clearly orthogonal matrices.
Therefore, $\varsigma_{1,\Upsilon}$ is a singular value of $\Upsilon$,
and thus \eqref{eq:condition3e} and condition (iii) imply that 
\begin{subequations}
\begin{align}
&~\varsigma_{1,\Upsilon}^2= 
\frac{ \nu_{\varepsilon}^2 \varsigma_1^2 }{\nu_x-  \varsigma_1^2} \leq 1
\label{eq:condition3fa}
\\
\Leftrightarrow &~ 
  \varsigma_1^2 \leq 
(\mu^{-1} \gamma\sigma_{\varepsilon}^{-2} -1)
( \mu^{-1}\gamma\sigma_x^{-2} -  \varsigma_1^2).
\label{eq:condition3fb} 
\end{align}
\end{subequations}
After a simple manipulation of \eqref{eq:condition3fb} 
under conditions (i) and (ii) with $\varsigma_1^2=\lambda_{\max}(\inputA^\top\inputA)$,
we obtain \eqref{eq:mu_condition_stable_regression}.

Conversely, suppose that \eqref{eq:mu_condition_stable_regression}
holds.
Then, conditions (i) and (ii) hold immediately,
and it is therefore sufficient to inspect condition (iii).
It is clear that \eqref{eq:mu_condition_stable_regression}
implies the inequality in \eqref{eq:condition3fa}.
Since $\nu_{\varepsilon}^2 \varsigma^2/(\nu_x- \varsigma^2)$ is an increasing
function of $\varsigma^2\in[0,\nu_x)$,
\eqref{eq:condition3fa} implies that
\begin{equation}
\varsigma_{i,\Upsilon} := 
\frac{ \nu_{\varepsilon} \varsigma_i}{(\nu_x- \varsigma_i^2)^{1/2}} \in(0,1],
~\forall \varsigma_i>0.
\label{eq:condition3g} 
\end{equation}
Let $\varsigma_{i,\Upsilon} := 0$ for all $\varsigma_i=0$ if any.
Define a diagonal matrix $\Sigma_\Upsilon\in\real^{n\times m}$,
in the same way as above, with diagonal entries $\varsigma_{i,\Upsilon}$.
Redefine the matrices
$\Upsilon:=V\Sigma_\Upsilon (-U)^\top$
and $\tilde{\Upsilon}:=V\Sigma_\Upsilon U^\top$.
Then, $\Upsilon$ and $\tilde{\Upsilon}$ have
the singular values $\varsigma_{i,\Upsilon}\in[0,1]$,
$i\in\{1,2,\cdots,\min\{n,m\}\}$.
Since
$\tilde{\Upsilon}$ satisfies $\eqref{eq:condition3b}$
and thus $\eqref{eq:condition3a}$,
$\Upsilon$ satisfies the equation of condition (iii).
\migip

\setcounter{appnum}{3}
\setcounter{equation}{0}

\section{Proof of Lemma \ref{lemma:compositionA}}
\label{subsec:proof_lemma_compositionA}

Let $\mathcal{V}^{\perp}\subset\hilbertx$ denote
the orthogonal complement of $\mathcal{V}$.
Then, it follows that
\begin{align}
\hspace*{-1.5em}^\gamma(\psi\circ\mathfrak{A})(x)
= &~ \min_{u\in\hilbertx}~ \big[
\psi(\mathfrak{A} u) + 0.5\gamma^{-1} \norm{u-x}^2
\big]
\nonumber\\
= &~\min_{u\in\hilbertx} ~\big[
\psi(\genopL u+b) + 
0.5\gamma^{-1} (
\norm{P_{\mathcal{V}}u-P_{\mathcal{V}}x}^2
\nonumber\\
 &\hspace*{2.8em}+
\norm{P_{\mathcal{V}^{\perp}}u- P_{\mathcal{V}^{\perp}}x}^2)
\big]
\nonumber\\
= &~\min_{u\in\hilbertx} ~\big[
\psi(\genopL u+b) + 
0.5\gamma^{-1}
\norm{P_{\mathcal{V}}u-P_{\mathcal{V}}x}^2\big]
\nonumber\\
= &~\min_{u\in\hilbertx} ~\big[
\psi(\genopL u+b) + 
0.5\gamma^{-1}
\norm{\genopL u - \genopL x}^2\big]
\nonumber\\
= &~\min_{\newz\in\hilbertarbk} ~\big[
\psi(\newz+b) + 
0.5\gamma^{-1}
\norm{\newz- \genopL x}^2\big]
\nonumber\\
= &~\min_{\newv\in\hilbertarbk} ~\big[
\psi(\newv) + 
0.5 \gamma^{-1}
\norm{\newv - \mathfrak{A} x}^2\big]
\nonumber\\
= &~ ^{\gamma}\psi(\mathfrak{A} x).
\label{eq:phiAmoreau}
\end{align}
Here, the second equality is due to the Pythagorean theorem,
the third equality holds because 
$\psi(\genopL u + b)$ is independent of $P_{\mathcal{V}^{\perp}} u$,
 the fourth equality is due to
$\genopL^* \genopL = P_{\mathcal{V}}=P_{\mathcal{V}}^* \circ P_{\mathcal{V}}$,
and finally the fifth equality is due to
$\range \genopL = \hilbertarbk$.
By \eqref{eq:phiAmoreau}, it follows that
$ (\psi\circ \mathfrak{A})_{\gamma^{-1/2} I}(x)  =
\psi(\mathfrak{A} x ) - ~^\gamma(\psi\circ\mathfrak{A})(x)=
(\psi -\hspace*{.1em}^\gamma\psi) (\mathfrak{A} x )$, which completes
 the proof.
\migip


 \begin{biography}{Masahiro Yukawa}
received
the B.E., M.E., and Ph.D. degrees from the Tokyo
Institute of Technology in 2002, 2004, and 2006,
respectively. 
He is a Professor with the Department
of Electronics and Electrical Engineering, Keio University, Yokohama, Japan.
He is currently a Senior Area Editor of the IEEE Transactions on Signal Processing.
He served as an Associate Editor for the IEEE Transactions on Signal
 Processing from 2015 to 2019, the Springer Journal of Multidimensional
Systems and Signal Processing from 2012 to 2016, and the IEICE Transactions
on Fundamentals of Electronics, Communications and Computer Sciences from
2009 to 2013. His research interests include mathematical adaptive signal
processing, convex/sparse optimization, and machine learning. 

Dr.~Yukawa received 
the JSPS Prize in 2021,
the Young Scientists' Prize, the Commendation for Science and Technology by
the Minister of Education, Culture, Sports, Science and Technology in 2014,
the Excellent Paper Award from the IEICE in 2006, among many others.
He is a Member of the IEICE.
 \end{biography}

 \begin{biography}{Hiroyuki Kaneko}
 received the B.E. and M.E. degrees in Electronics and Electrical
  Engineering from Keio University, Yokohama, Japan, in 2020 and 2022,
  respectively.
 He is currently a Researcher with NTT Communication Science Laboratories, NTT Corporation, Kyoto, Japan.
 His research interests include sparse signal processing, convex optimization, and audio signal processing.
 \end{biography}

 \begin{biography}{Kyohei Suzuki}
  (Student Member, IEEE) received the B.E.~and M.E.~degrees in
  Electronics and Electrical Engineering from Keio University, Yokohama,
  Japan, in 2020 and 2022, respectively.
   He is currently working toward the Ph.D.~degree in Electronics and Electrical Engineering from Keio University, Yokohama, Japan.
   His research interests include mathematical signal processing, sparse optimization, and robust statistics.
 \end{biography}

 \begin{biography}{Isao Yamada}
received the B.E. degree in computer science from the University of Tsukuba, Tsukuba,
Japan, in 1985, and the M.E. and Ph.D. degrees in electrical and electronic engineering from the Tokyo Institute of Technology, Tokyo, Japan, in 1987
and 1990, respectively. He is currently a Professor with the Department of Information and Communications Engineering,
Tokyo Institute of Technology. His current research interests are in mathematical signal processing, nonlinear inverse problems, and optimization theory.
He has been the IEICE Fellow since 2015. He was the recipient of the MEXT Minister Award (Research Category), the IEEE Signal Processing Magazine Best Paper Award in 2015,
the IEICE Excellent Paper Awards (in 1991, 1995, 2006, 2009, 2014 and 2022), the IEICE Achievement Award in 2009, the ICF Research Award in 2004,
the Docomo Mobile Science Award (Fundamental Science Division) in 2005 and the Fujino Prize in 2008. 
He served as a member of the IEEE Signal Processing Society Awards Board in 2022.
 \end{biography}

\end{document}